\documentclass[%
superscriptaddress,
 amsmath,amssymb,
 aps, physrev,
twocolumn,
]{revtex4-2}

\usepackage{graphicx}
\usepackage{dcolumn}
\usepackage{bm}
\usepackage{hyperref}

\usepackage{xcolor}
\usepackage{physics}

\usepackage{amsthm}

\newtheorem{theorem}{Theorem}

\newtheorem{lemma}{Lemma}
\newtheorem{proposition}{Proposition}

\theoremstyle{definition}
\newtheorem{definition}{Definition}

\newcommand{\dist}[0]{\mathrm{dist}}
\newcommand{\rrangle}{\rangle\!\rangle}
\newcommand{\llangle}{\langle\!\langle}
\newcommand{\Ket}[1]{| #1 \rrangle}
\newcommand{\Bra}[1]{\llangle #1 |}

\begin{document}

\preprint{APS/123-QED}

\title{\textbf{Classical simulation of noisy random circuits from exponential decay of correlation} 
}%

\author{Su-un Lee}%
\email{suun@uchicago.edu}
\affiliation{Pritzker School of Molecular Engineering, The University of Chicago, Chicago, IL 60637, USA}

\author{Soumik Ghosh}
\email{soumikghosh@uchicago.edu}
\affiliation{Department of Computer Science, University of Chicago, Chicago, Illinois 60637, USA}

\author{Changhun Oh}
\affiliation{Department of Physics, Korea Advanced Institute of Science and Technology, Daejeon 34141, Republic of Korea}

\author{Kyungjoo Noh}
\affiliation{AWS Center for Quantum Computing, Pasadena, CA 91125, USA}

\author{Bill Fefferman}
 \email{wjf@uchicago.edu}
\affiliation{Department of Computer Science, University of Chicago, Chicago, Illinois 60637, USA}

\author{Liang Jiang}
 \email{liang.jiang@uchicago.edu}
\affiliation{Pritzker School of Molecular Engineering, The University of Chicago, Chicago, IL 60637, USA}

\date{\today}

\begin{abstract}
    We study the classical simulability of noisy random quantum circuits under general noise models. While various classical algorithms for simulating noisy random circuits have been proposed, many of them rely on the anticoncentration property, which can fail when the circuit depth is small or under realistic noise models. We propose a new approach based on the exponential decay of conditional mutual information (CMI), a measure of tripartite correlations. We prove that exponential CMI decay enables a classical algorithm to sample from noisy random circuits---in polynomial time for one dimension and quasi-polynomial time for higher dimensions---even when anticoncentration breaks down. To this end, we show that exponential CMI decay makes the circuit depth effectively shallow, and it enables efficient classical simulation for sampling. We further provide extensive numerical evidence that exponential CMI decay is a universal feature of noisy random circuits across a wide range of noise models. Our results establish CMI decay, rather than anticoncentration, as the fundamental criterion for classical simulability, and delineate the boundary of quantum advantage in noisy devices.
\end{abstract}

\maketitle

\section{Introduction}

Demonstrating quantum advantage using near-term quantum devices remains a central goal in quantum information science. These devices, often characterized by a limited number of qubits and nonnegligible noise, do not yet support fault-tolerant quantum computation~\cite{Preskill2018quantumcomputingin}. Nevertheless, researchers aim to identify specific computational tasks where these noisy intermediate-scale quantum (NISQ) devices can outperform classical computation. Among the prominent proposals is random circuit sampling (RCS), a task conceived to be intractable for classical computation under plausible complexity-theoretic assumptions~\cite{Aaronson2011computational,Bouland2019OnComplexityVerificationQuantum2019}, and have demonstrated by landmark experiments~\cite{Arute2019quantumsupremacy, MorvanPhaseTransitionsRandom2024, Zhu2022quantumcomputationaladvantage, DeCrossComputationalPowerRandom2025}.

However, the classical hardness of RCS in the presence of noise remains unclear. Recent advances in classical algorithms show that noisy RCS can be efficiently simulated classically if the output distribution remains sufficiently ``flat,'' i.e., it exhibits anticoncentration~\cite{AharonovPolynomialTimeClassical2023}. Moreover, it has been shown that noisy RCS under certain entropy-increasing (unital) noise channels exhibits anticoncentration above logarithmic depth~\cite{deshpande2022tightboundsnoisy, dalzell2022randomquantumanticoncentrate}. Yet the classical simulability of noisy RCS at shallow depths or under general noise models is still open. For shallow-depth circuits, it is conjectured that noiseless RCS becomes classically hard once the depth exceeds a constant threshold $d^*=\Theta(1)$~\cite{nappEfficientClassicalSimulation2022, mcginley2025measurementinduced, watts2025quantumadvantage}. Since such circuits can be realized with relatively high fidelity, one might expect shallow-depth noisy RCS to retain quantum advantage against moderate noise levels~\cite{cheng2023efficientsampling}. In addition, non-unital noise---ubiquitous in quantum devices due to $T_1$ decay and readout bias (e.g.,~\cite{Arute2019quantumsupremacy, MorvanPhaseTransitionsRandom2024, Zhu2022quantumcomputationaladvantage, DeCrossComputationalPowerRandom2025})---reduces entropy and drastically alters the output distribution. Indeed, Ref.~\cite{FeffermanEffectNonunitalNoise2024} shows that non-unital noise can break anticoncentration, thereby invalidating existing simulation frameworks~\cite{AharonovPolynomialTimeClassical2023, schuster2024polynomialtimeclassicalalgorithmnoisy, aharonov1996polynomialsimulationsdecoheredquantum}.

Meanwhile, exponential decay of correlation is a central phenomenon in many-body physics with both fundamental and practical implications~\cite{Brand_o_2014, Bluhm2022exponentialdecayof}. A widely accepted measure of such correlation is the conditional mutual information (CMI), which captures tripartite dependencies between subsystems. Physical intuition and prior results suggest that CMI between distant regions is unstable under local perturbations and decays exponentially in many physically relevant settings~\cite{KatoQuantumApproximateMarkovChains2019,brown2012quantummarkovnetworkscommuting,Jouneghani2014investigationcommutinghamiltonian,kuwahara2024clusteringconditionalmutualinformation, chen2025quantumgibbsstateslocally, kato2025clusteringconditionalmutualinformation}. Motivated by these observations, CMI decay is often taken to be a key assumption in proposals for topological phases of matter~\cite{ShiEntanglementBootstrapApproach2021, yang2025topologicalmixedstatesaxiomatic, ShiDomainWallTopologicalEntanglement2021, SangStabilityMixedStateQuantum2025, negari2025spacetimemarkovlengthdiagnostic}. From a computational standpoint, Ref.~\cite{nappEfficientClassicalSimulation2022} proved that assuming exponential CMI decay, there exists an efficient classical algorithm to sample from the output distribution of a 2D shallow-depth quantum circuit. However, the runtime of that algorithm scales exponentially with the circuit depth and thus does not suffice to simulate deep random circuits.

In this work, we incorporate these physical insights into RCS, and show that exponential decay of CMI implies classical simulability for noisy RCS, regardless of the specific noise model and circuit depth. Our approach relies on two key steps: (i) we rigorously establish that exponential CMI decay implies that the output distributions can be approximated by those of shallow circuits of depth $O(\log^D n)$ for $D$-dimensional random circuits; and (ii) we prove that such shallow circuits with decaying CMI can be efficiently simulated classically---polynomial time in 1D and quasi-polynomial time in higher dimensions. Importantly, our algorithm applies to any noise channel that yields exponential CMI decay, even in regimes where anticoncentration fails.

Although a general proof of CMI decay for arbitrary noise models remains open, we provide extensive numerical evidence supporting its universality in noisy random circuits. Specifically, we observe exponential CMI decay in 1D Haar-random circuits and 2D Clifford-random circuits across a wide range of noise models, including both unital and non-unital channels. These findings indicate that exponential CMI decay is an intrinsic property of noisy random circuits and provide strong evidence for the classical simulability of RCS in noisy devices.

\section{Problem setup}

\begin{figure}
    \centering
    \includegraphics[width=\linewidth]{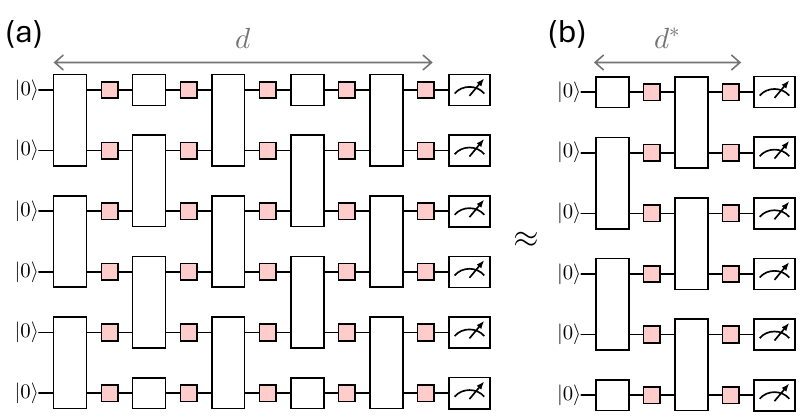}
    \caption{(a) A depth-$d$ noisy random circuit consists of alternating layers of unitary gates (white boxes) and noise channels (red boxes) (b) Assuming exponential decay of CMI, the output distribution can be approximated by a depth-$d^* = O(\log^D (n/\varepsilon))$ circuit.}
    \label{fig:circuits}
\end{figure}

We consider a noisy random circuit on $n$ qubits arranged on a $D$-dimensional lattice $\Lambda$, initialized in the product state $\ket{0^n}$. The circuit consists of $d$ alternating layers of two-qubit unitary gates and single-qubit noise channels. Each unitary layer $U_i$ applies two-qubit gates to disjoint neighboring pairs in $\Lambda$, drawn independently from a distribution forming at least a unitary 2-design (e.g., the Haar measure on $U(4)$). After each unitary layer, every qubit undergoes a single-qubit noise channel $\mathcal{N}$, leading to the overall evolution described by
\begin{equation}
\mathcal{C} = \mathcal{N}^{\otimes n} \circ \mathcal{U}_d \circ \dots \circ \mathcal{N}^{\otimes n} \circ \mathcal{U}_1,
\label{eq:noisy_random_circuit}
\end{equation}
where $\mathcal{U}_i(\rho) = U_i \rho U_i^\dagger$. The final state $\mathcal{C}(\ketbra{0^n}{0^n})$ is measured in the computational basis [Fig.~\ref{fig:circuits}(a)]. We assume that the noise channel $\mathcal{N}$ is fixed (independent of system size), corresponding to a constant noise rate, and is not unitary, i.e., $\mathcal{N}(\cdot) \ne U(\cdot)U^\dagger$ for any unitary $U$. This includes non-unital noise models such as amplitude damping and reset channels, which capture relevant physical mechanisms in current quantum devices.

The goal of our classical algorithm is to approximately sample from the output distribution $P:\{0,1\}^n \to [0,1]$ of the state $\mathcal{C}(\ketbra{0^n}{0^n})$ in small total variation distance. The correlational structure of $P$, characterized by CMI, plays a central role in enabling our algorithm. For subsets $X, Y, Z \subset \Lambda$, the CMI is defined as
\begin{equation}
I(X:Z|Y) = H(XY) + H(YZ) - H(XYZ) - H(Y),
\end{equation}
where $H(X) = -\sum_{x \in \{0,1\}^{|X|}} P_X(x)\log_2 P_X(x)$ denotes the Shannon entropy of the marginal distribution $P_X$. We now define the \emph{approximate Markov condition}, adapting the formulation from Ref.~\cite{brandaoFiniteCorrelationLength2019} to our setting:

\begin{definition}
    Let $P$ be the output distribution of a quantum circuit on $\Lambda$. For a function $\eta : \mathbb{R}_{\ge 0} \rightarrow \mathbb{R}_{\ge 0}$, we say that $P$ satisfies the \emph{$\eta(\ell)$-approximate Markov condition} if, for any tripartition $X \sqcup Y \sqcup Z \subset \Lambda$ with $\dist(X, Z) \ge \ell$,
    \begin{equation}
        I(X:Z|Y) \le \eta(\ell).
    \end{equation}
    If $P$ arises as the output distribution of a random quantum circuit $\mathcal{C}$, we say that it satisfies the \emph{average $\eta(\ell)$-approximate Markov condition} if, for all such tripartitions,
    \begin{equation}
        \mathbb{E}_{\mathcal{C}} \, I(X:Z|Y) \le \eta(\ell),
    \end{equation}
    where the expectation is taken over the choice of $\mathcal{C}$.
\end{definition}

Here, $\dist(X,Z)$ is defined as $\min_{x\in X, z \in Z}\dist(x,z)$ where $\dist(x,z)$ is the shortest path between $x, z$ in $\Lambda$. This condition implies that correlations between $X$ and $Z$, conditioned on $Y$, decay with separation. Specifically, Pinsker's inequality~\cite{coverElementsInformationTheory2006} yields
\begin{equation}
    I(X:Z|Y) \ge \frac{1}{2\ln2} \| P_{XYZ} - P_X P_{Y|X} P_{Z|Y} \|_1^2,
    \label{eq:pinsker}
\end{equation}
indicating that small CMI implies approximate conditional independence. Here, $\|\cdot\|_1$ denotes the $\ell^1$ norm (twice the total variation distance), and $P_{Y|X}(y|x) = P_{XY}(x,y)/P_X(x)$ is the conditional distribution of $Y$ given $X$, and similarly for $P_{Z|Y}$.

\begin{figure}
    \centering
    \includegraphics[width=\linewidth]{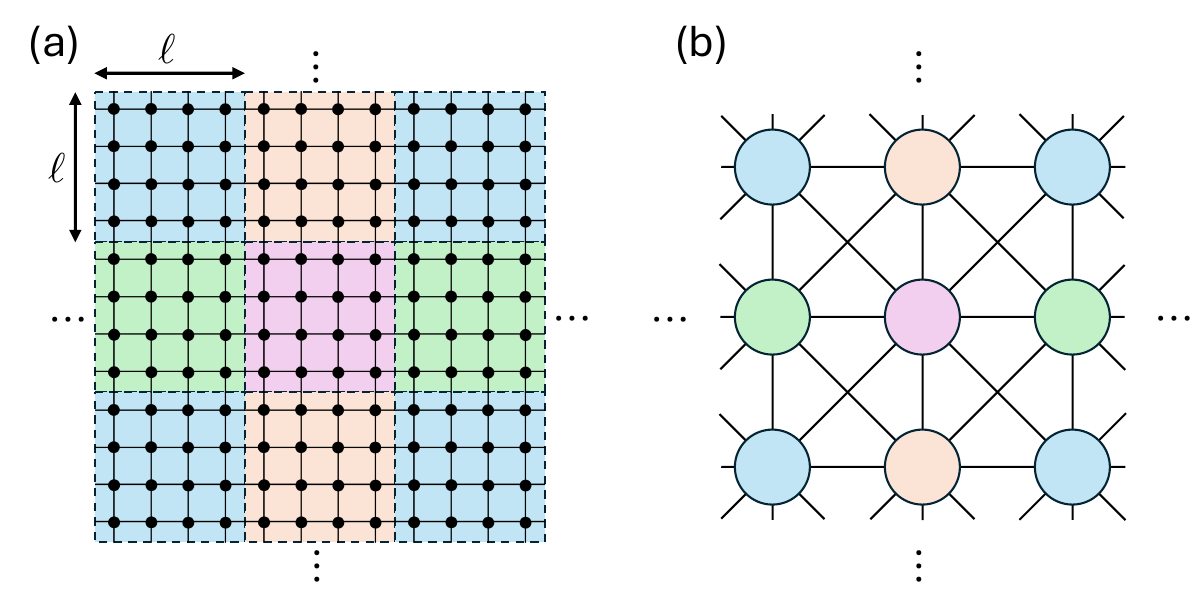}
    \caption{Example of coarse-graining a 2D grid. (a) A 2D grid is partitioned into disjoint squares of side length $\ell$. (b) The coarse-grained graph $G = (V, E)$, where each vertex corresponds to a hypercube of side length $\ell$, and edges are drawn between two vertices if the distance between them is less than $\ell$.}
    \label{fig:2Dpatches}
\end{figure}

To analyze distributions satisfying the $\eta(\ell)$-approximate Markov condition, we coarse-grain the lattice $\Lambda$ into a graph $G = (V, E)$ as follows. We partition $\Lambda$ into disjoint hypercubes of side length $\ell$, each containing $\ell^D$ qubits, and denote the set of hypercubes by $V$. For each pair $X, Y \in V$, we draw an edge $\{X, Y\} \in E$ if $\dist(X, Y) < \ell$. This coarse-grained structure defines a graph $G$ whose connectivity reflects short-range correlations in the original circuit [Fig.~\ref{fig:2Dpatches}].

Any probability distribution $P$ over $\Lambda$ can be decomposed, under an ordering $V = {X_1, X_2, \dots, X_{|V|}}$, into a chain of conditional distributions:
\begin{equation}
    P = \prod_{i=1}^{|V|} P_{X_i| X_{<i}},
\end{equation}
where $X_{<i} = \bigcup_{j=1}^{i-1}X_j$. Sampling from $P$ in this form requires conditioning on the full history $X_{<i}$, which is generally intractable. However, the coarse-grained graph $G = (V, E)$ provides simplifications: (i) by construction, any two non-neighboring partitions in $V$ are at least $\ell$ apart, and (ii) each partition $X \in V$ has a constant number of neighbors: denoting $N(X)=\{Y\in V: \{X, Y\}\in E\}$, we have $|N(X)| = 3^D - 1$. These properties allow the global conditional to be approximated by a local conditional, yielding the following:

\begin{proposition}
    Let $P$ be a probability distribution that satisfies the $\eta(\ell)$-approximate Markov condition. Then
    \begin{equation}
        \left\| P - \prod_{i=1}^{|V|} P_{X_i|N'(X_i)} \right\|_1 \le O(n/\ell^D) \sqrt{\eta(\ell)},
    \end{equation}
    where $N'(X_i) = N(X_i) \cap X_{<i}$.
    \label{prop:patches}
\end{proposition}

\section{Main results}

Our main result shows that if the output distribution of a noisy random circuit satisfies an exponentially decaying CMI (i.e., approximate Markov condition), then it can be classically simulated in polynomial time for one-dimensional systems and quasi-polynomial time for higher dimensions. This result relies on two key ingredients: (i) the approximate Markov condition implies that a noisy circuit of arbitrary depth can be approximated by one of depth $d^* = O(\log^D(n/\varepsilon))$; and (ii) the approximate Markov condition also ensures that such shallow circuits can be simulated classically. Together, these results establish that the output distribution of a noisy random circuit can be sampled classically with high accuracy.

We begin by establishing that the approximate Markov condition implies that a deep noisy circuit can be effectively approximated by a shallow one.

\begin{theorem}
    Let $\rho$ and $\sigma$ be arbitrary density matrices over the qubits in a $D$-dimensional grid $\Lambda$, and let $\mathcal{C}$ be a depth-$d$ noisy random circuit. Denote the output distributions of $\mathcal{C}(\rho)$ and $\mathcal{C}(\sigma)$ by $P$ and $Q$, respectively. Suppose both distributions satisfy the average $\mathrm{poly}(n)\exp(-\Omega(\ell))$-approximate Markov condition for all $\ell$. Then, for any $\varepsilon > 0$, there exists $d^* = O(\log^D(n/\varepsilon))$ such that
    \begin{equation}
        \mathbb{E}_{\mathcal{C}}\|P-Q\|_1  \le \varepsilon,
    \end{equation}
    for all $d \ge d^*$.
    \label{thm:effective_shallow_depth}
\end{theorem}

This directly implies that the output distribution of a deep noisy circuit can be approximated by that of a shallow one. To see this, we decompose $\mathcal{C}$ as $\mathcal{C} = \mathcal{C}_2 \circ \mathcal{C}_1$, where $\mathcal{C}_1$ consists of the first $d - d^*$ layers and $\mathcal{C}_2$ consists of the final $d^*$ layers:
\begin{align}
    \mathcal{C}_1 &= \mathcal{N}^{\otimes n} \circ \mathcal{U}_{d-d^*} \circ \cdots \circ \mathcal{N}^{\otimes n} \circ \mathcal{U}_{1}, \\
    \mathcal{C}_2 &= \mathcal{N}^{\otimes n} \circ \mathcal{U}_{d} \circ \cdots \circ \mathcal{N}^{\otimes n} \circ \mathcal{U}_{d - d^* + 1}.
\end{align}
Let $\rho = \mathcal{C}_1(\ketbra{0^n}{0^n})$ be the intermediate state after applying $\mathcal{C}_1$, and define $P = \mathcal{C}_2(\rho)$ and $Q = \mathcal{C}_2(\ketbra{0^n}{0^n})$. If both $P$ and $Q$ satisfy ${\rm poly}(n)\exp(-\Omega(\ell))$-approximate Markov condition averaged over $\mathcal{C}_2$, then Theorem~\ref{thm:effective_shallow_depth} implies that $\mathbb{E}_{\mathcal{C}}\|P - Q\|_1 \le \varepsilon$, for $d^* = O(\log^D(n/\varepsilon))$ [Fig.~\ref{fig:circuits}(b)]. A sketch of the proof is given below, leaving the complete proof to Appendix~\ref{sec:proof-theorem1}.

\begin{proof}[Sketch of Proof]
    The proof proceeds by combining two key observations. First, under the ${\rm poly}(n)\exp(-\Omega(\ell))$-approximate Markov condition for both $P$ and $Q$, Proposition~\ref{prop:patches} implies that for $\ell = O(\log(n))$, both distributions are approximated as
    \begin{equation}
        P \approx \prod_{i=1}^{|V|} P_{X_i | N'(X_i)}, \quad Q \approx \prod_{i=1}^{|V|} Q_{X_i | N'(X_i)}.
    \end{equation}
    This structure allows us to bound the distance $\|P - Q\|_1$ in terms of local marginals:
    \begin{multline}
        \|P - Q\|_1 \lessapprox \sum_{i=1}^{|V|} \big(
            \|P_{X_i \cup N'(X_i)} - Q_{X_i \cup N'(X_i)}\|_1 \\
            + \|P_{N'(X_i)} - Q_{N'(X_i)}\|_1
    \big).
    \label{eq:bounding_from_marginals_main}
    \end{multline}

    Second, we show that the differences between corresponding marginals $\|P_X - Q_X\|_1$ decay rapidly with circuit depth. In particular, the difference between expectation values of local observables for $\mathcal{C}(\rho)$ and $\mathcal{C}(\ketbra{0^n}{0^n})$ vanishes with depth~\cite{mele2024noiseinducedshallowcircuitsabsence}. This implies that local marginals become indistinguishable after depth $O(|X|)$.

    Since both $|X_i \cup N'(X_i)|$ and $|N'(X_i)|$ are $O(\ell^D)$, we conclude that $\|P - Q\|_1 \le \varepsilon$ for depth $d^* = O(\log^D(n/\varepsilon))$.
\end{proof}

Although the approximate Markov condition implies that noisy circuits can be reduced to shallow depth, shallow circuits alone are not necessarily classically simulable. Indeed, certain families of shallow-depth circuits are known or conjectured to be classically hard to simulate (see, e.g., \cite{GaoQuantumSupremacySimulating2017,BermejoVegaArchitecturesQuantumSimulation2018,HaferkampClosingGapsQuantum2020}). Nevertheless, we show that shallow-depth circuits satisfying the approximate Markov condition can be simulated efficiently. This is achieved by generalizing the “Patching” algorithm of Ref.~\cite{nappEfficientClassicalSimulation2022}.

\begin{theorem}
    Let $P$ be the output distribution of a depth-$d$ quantum circuit on a $D$-dimensional grid $\Lambda$. If $P$ satisfies the $\mathrm{poly}(n)\exp(-\Omega(\ell))$-approximate Markov condition for all $\ell$, then there exists a classical algorithm that outputs samples from a distribution $P'$ satisfying $\|P - P'\|_1 \le \varepsilon$, in runtime $n\cdot \exp\left({O(d \cdot (d + \log (n/\varepsilon))^{D-1})}\right)$.
    \label{thm:shallow-depth_algorithm}
\end{theorem}

We begin by approximating the target distribution $P$ with $P' = \prod_{i=1}^{|V|} P_{X_i|N'(X_i)}$ by choosing $\ell = O(\log(n/\varepsilon))$ (as justified by Proposition~\ref{prop:patches}). Our goal is to sample from $P'$, which factorizes over local conditional distributions. The algorithm proceeds iteratively as follows.

Suppose we have already sampled $X_1, \dots, X_{j-1}$ from $\prod_{i=1}^{j-1}P_{X_i|N'(X_i)}$. We then sample qubits in $X_j$ conditioned on the outcomes in $N'(X_j)$, denoted $x \in \{0,1\}^{M}$, where $M$ is the number of qubits contained in $N'(X_j)$. The conditional distribution $P_{X_j|N'(X_j) = x}$ is then computed by simulating the region $X_j \cup N'(X_j)$ along with its lightcone [Appendix~\ref{sec:preliminaries}], which is contained within a $D$-dimensional hypercube containing at most $(3\ell + 2d)^D$ qubits. This simulation step can be implemented using standard tensor contraction techniques (e.g., Corollary 1.5 in~\cite{MarkovShiSimulatingQuantumComputation2008}) in time $\exp(O(d(\ell + 2d)^{D-1}))$. 

After obtaining $x_j \sim P_{X_j|N'(X_j)=x}$, we proceed to the next partition $X_{j+1}$. Once all $X_i$ have been sampled, we obtain a full sample from $P'$. Since there are $O(n/\ell^D)$ partitions, the total runtime is $n/\ell^D \cdot \exp(O(d(\ell + 2d)^{D-1}))$, which establishes Theorem~\ref{thm:shallow-depth_algorithm}.

Finally, we combine Theorems~\ref{thm:effective_shallow_depth} and~\ref{thm:shallow-depth_algorithm} to establish our main result. Theorem~\ref{thm:effective_shallow_depth} guarantees that, for $d^* = O(\log^D(n/\varepsilon))$, the output distribution $P$ of a noisy random circuit can be approximated by the output distribution $P'$ of a depth-$d^*$ circuit. Furthermore, since $P'$ also supposed to satisfy the approximate Markov property, Theorem~\ref{thm:shallow-depth_algorithm} ensures that $P'$ can be sampled classically in time $\exp\left(O\left(\log^{D^2} (n/\varepsilon)\right)\right)$. Deferring the formal proof to Appendix~\ref{sec:proof-theorem3}, we summarize our main result as follows:

\begin{theorem}[Main Theorem]
    Let $\mathcal{C}$ be a depth-$d$ noisy random circuit on a $D$-dimensional grid $\Lambda$, and $P$ be the output distribution of $\mathcal{C}(\ketbra{0^n}{0^n})$. Suppose $P$ satisfies the average ${\rm poly}(n)\exp(-\Omega(\ell))$-approximate Markov condition for all $\ell$ and $d$. Then there exists a classical algorithm that outputs a sample from $P'$ such that $\|P-P'\|_1 \le \varepsilon$ with probability at least $1-\delta$ over the choice of $\mathcal{C}$, in runtime ${\rm poly}(n, 1/\varepsilon, 1/\delta)$ for $D=1$ and ${\rm quasipoly}(n, 1/\varepsilon, 1/\delta)$ for $D \ge 2$.
    \label{thm:main_theorem}
\end{theorem}

\begin{figure*}[t]
    \centering
    \includegraphics[width=\linewidth]{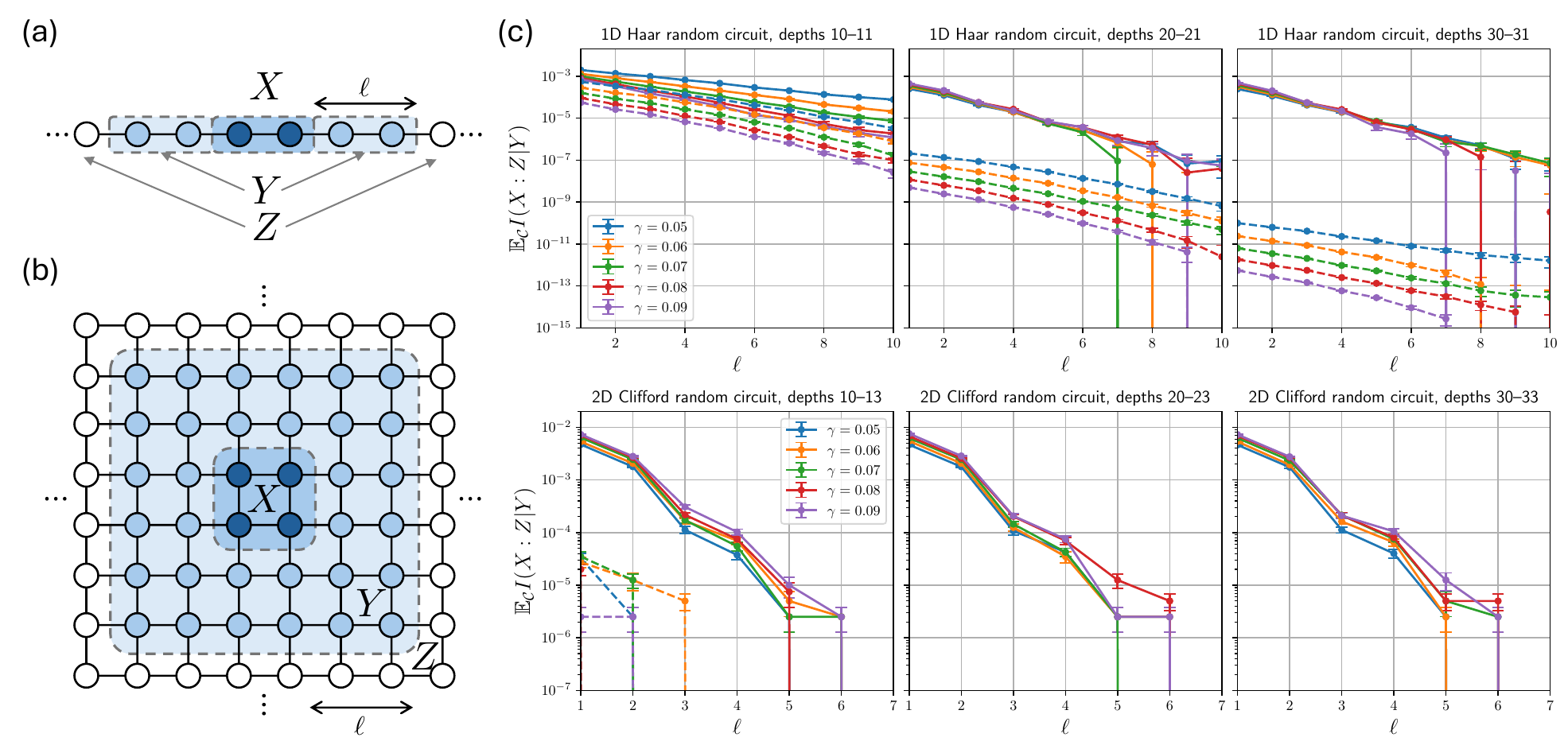}
    \caption{(a) Simulation setup of the 1D Haar random circuits. (b) Simulation setup of the 2D Clifford random circuits. (c) Simulation results of the CMI decay in 1D Haar random circuits (top row) and 2D Clifford random circuits (bottom row). For 1D Haar random circuits, CMI is averaged over $64$ circuit realizations, and CMI is estimated with $1{,}000$ MC samples. For 2D Clifford random circuits, CMI is averaged over $100{,}000$ circuit realizations, and CMI is exactly computed. To reduce fluctuations due to specific gate patterns, we average the CMI over pairs of consecutive depths in 1D (e.g., depths 10--11) and over groups of four consecutive depths in 2D (e.g., depths 10--13). For all plots, solid lines present the results with non-unital noise channels (amplitude damping and heralded reset channel), while dashed lines present the results with unital noise channels (depolarizing channel and heralded depolarizing channel).}
    \label{fig:numerical}
\end{figure*}

\section{Numerical evidence of decaying CMI}

While we have shown that exponential decay of CMI enables classical simulation of noisy random circuits, we provide extensive numerical evidence that such decay indeed occurs in the output distributions across a broad range of noise models. Specifically, we simulate 1D Haar-random circuits using the matrix product state (MPS) method~\cite{Vidal2003EfficientClassicalSimulation, VerstraeteMatrixProductDensityOperators2004, Zwolak2004MixedStateDynamics,Noh2020efficientclassical}, and 2D Clifford-random circuits using stabilizer simulation~\cite{gottesman1998heisenbergrepresentationquantumcomputers,AaronsonGottesmanImprovedSimulationStabilizer2004}.

In the 1D Haar-random circuit setup, $32$ qubits are arranged in a line. Two-qubit Haar-random gates are applied in parallel to even (resp. odd) neighboring pairs, followed by a single-qubit noise channel on each qubit. We examine two noise models: the non-unital amplitude damping channel,
\begin{equation}\label{eq:amplitude_damping}
    \mathcal{N}_{\rm amp}(\cdot) = K_0 (\cdot) K_0^\dagger + K_1 (\cdot) K_1^\dagger,
\end{equation}
with $K_0 = \begin{pmatrix} 1 & 0 \\ 0 & \sqrt{1-\gamma} \end{pmatrix}$ and $K_1 = \begin{pmatrix} 0 & \sqrt{\gamma} \\ 0 & 0 \end{pmatrix}$, and the unital depolarizing channel,
\begin{equation}
    \mathcal{N}_{\rm depo}(\rho) = (1-\gamma) \rho + \gamma \frac{I}{2},
\end{equation}
where $\gamma \in [0,1]$ denotes the noise rate.

The 2D Clifford random circuits consist of $32\times 32$ qubits arranged in a square grid where we apply two-qubit Clifford gates on disjoint pairs of neighboring qubits in parallel, followed by a single-qubit noise channel. We consider two noise models compatible with stabilizer simulation: the heralded reset channel and the heralded depolarizing channel,
\begin{align}
    \mathcal{N}_{\rm hreset}(\rho) &= \begin{cases}
    \ketbra{0}{0}, & \text{with probability } \gamma,\\
    \rho, & \text{with probability } 1-\gamma,
    \end{cases}\\
    \mathcal{N}_{\rm hdepo}(\rho) &= \begin{cases}
    I/2, & \text{with probability } \gamma,\\
    \rho, & \text{with probability } 1-\gamma,
    \end{cases}
\end{align}
where the heralded reset channel is non-unital, and the heralded depolarizing channel is unital.

In these simulations, we compute the CMI of the output distribution $P$ of the final state $\mathcal{C}(\ketbra{0^n}{0^n})$. For 1D Haar-random circuits (resp. 2D Clifford-random circuits), we define the region $X$ as two qubits in the middle (resp. a $2 \times 2$ square patch at the center), the region $Z$ as all qubits at least $\ell$-apart from $X$, and the remainder as $Y$ [Fig.~\ref{fig:numerical}(a) and (b)]. We then compute the conditional mutual information $I(X : Z | Y)$, averaging over many circuit realizations. Further simulation details are provided in Appendix~\ref{sec:simulation_details}.

Fig~\ref{fig:numerical}(c) shows the simulation results. For non-unital noise channels, we observe a clear exponential decay of CMI with distance $\ell$ in both 1D Haar-random and 2D Clifford-random circuits, across all circuit depths $d$ and noise rates $\gamma$. For unital noise channels, exponential decay is also evident in 1D Haar-random circuits; however, in 2D Clifford-random circuits, the CMI values are too small to be estimated reliably. These results support the assumption that noisy random circuits generically exhibit exponentially decaying CMI, thereby validating the classical simulation algorithm developed in this work~\footnote{We have further confirmed the exponential decay of CMI for various sizes of the region $X$; see Appendix~\ref{sec:additional_numerics}.}.

\section{Discussion}


While our analysis focuses on $D$-dimensional grid circuits, it naturally extends to more general architectures in which the underlying interaction graph has $\ell$-balls and $\ell$-local treewidth both growing at most polynomially with $\ell$ (see Appendix~\ref{sec:extended_results}). This includes a broad class of locally structured systems. However, our simulation algorithm does not apply to random circuits with all-to-all connectivity, such as those used in recent trapped-ion experiments~\cite{DeCrossComputationalPowerRandom2025}.

We also remark that while Theorem~\ref{thm:effective_shallow_depth} establishes that noisy random circuits become effectively shallow with depth $d^* = \mathrm{polylog}(n)$, we conjecture that this can be improved to $d^* = O(\log n)$. This is supported by known extremes: Ref.~\cite{shtanko2024complexitylocalquantumcircuits} shows that when each layer of the random circuit consists of global Haar-random gate, the effective depth is $O(1)$, while a simple analysis in Appendix~\ref{sec:effective_depth} shows that if the layers consist of single-qubit Haar-random gates, the effective depth is $O(\log n)$. Since our two-qubit gate model interpolates between these two limits, we expect a logarithmic depth to be sufficient.

Although we observe exponential CMI decay numerically across a variety of noisy random circuits, a general analytical proof remains as a challenge. One key difficulty is that CMI can grow under noise, contrary to the intuition that noise always suppresses correlations. While Ref.~\cite{yangWhenCanClassical2024} establishes CMI decay for shallow, noiseless 1D circuits, recent works~\cite{LeeUniversalSpreading2024, ZhangNonlocalCMI2024} demonstrate that noise can induce rapid spreading of CMI. These results highlight the subtle interplay between noise and correlation dynamics, underscoring the difficulty of establishing CMI decay analytically even in one dimension.

Finally, we remark that while our focus has been on random circuits, exponential decay of CMI may also emerge in more general, non-random circuit families, enabling efficient classical simulation. In particular, concurrent works~\cite{FrankSuun, JonJoel} study the conditional independence structure of arbitrary (worst-case) quantum circuits subject to depolarizing noise. These works prove that such circuits become classically tractable when the noise rate exceeds a certain threshold, regardless of circuit depth~\cite{FrankSuun}, and provide evidence of classical tractability when the circuit depth exceeds a certain threshold at any fixed constant noise rate~\cite{JonJoel}.

\begin{acknowledgments}
    We are grateful to Senrui Chen, Alexander Dalzell, Sarang Gopalakrishnan, Michael Gullans, Jon Nelson, Joel Rajakumar, Yihui Quek, and Yifan Zhang for their helpful discussions. S.L. especially thanks Gunhee Park for his help in numerical simulations.
    S.L. and L.J. acknowledge support from the ARO MURI (W911NF-21-1-0325), AFOSR MURI (FA9550-21-1-0209), NSF (ERC-1941583, OMA-2137642, OSI-2326767, CCF-2312755, OSI-2426975), and Packard Foundation (2020-71479). S.L. is partially supported by the Kwanjeong Educational Foundation.
    S.G. and B.F. acknowledge support from AFOSR (FA9550-21-1-0008). This material is based upon work partially supported by the National Science Foundation under Grant CCF-2044923 (CAREER), by the U.S. Department of Energy, Office of Science, National Quantum Information Science Research Centers (Q-NEXT).
    C.O. was supported by the National Research Foundation of Korea Grants (No. RS-2024-00431768 and No. RS-2025-00515456) funded by the Korean government (Ministry of Science and ICT~(MSIT)) and the Institute of Information \& Communications Technology Planning \& Evaluation (IITP) Grants funded by the Korea government (MSIT) (No. IITP-2025-RS-2025-02283189 and IITP-2025-RS-2025-02263264).
    Part of this research was performed while the author was visiting the Institute for Mathematical and Statistical Innovation (IMSI), which is supported by the National Science Foundation (Grant No. DMS-1929348).
    This work was completed with resources provided by the University of Chicago's Research Computing Center and NVIDIA academic grant program.
\end{acknowledgments}

\bibliography{references}

\appendix

\onecolumngrid

\section{\label{sec:preliminaries}Preliminaries}

We review preliminary concepts and results that will be useful for our main analysis.

\subsection{Basic properties of statistical distances}

The statistical distance (or, $\ell^1$ distance) between two probability distributions $P$ and $Q$ of a random variable $X$, which takes values in some finite set $\mathcal{X}$, is defined as
\begin{equation}
    \|P - Q\|_1 = \sum_{x \in \mathcal{X}} |P(x) - Q(x)|,
\end{equation}
which equals twice the total variation distance. Equivalently,
\begin{equation}
    \|P - Q\|_1 = 2 \cdot \max_{A \subset \mathcal{X}} |P(A) - Q(A)|,
\end{equation}
so $\|P-Q\|_1$ measures the maximal probability discrepancy that $P$ and $Q$ assign to the same event~\cite{coverElementsInformationTheory2006}.

A quantum analog of the statistical distance is the trace distance between two quantum states $\rho$ and $\sigma$, defined as
\begin{equation}
    \|\rho - \sigma\|_1 = \Tr\left(\sqrt{(\rho - \sigma)^\dagger (\rho - \sigma)}\right),
\end{equation}
which corresponds to the sum of the absolute values of the eigenvalues of $(\rho - \sigma)$. If the two states $\rho$ and $\sigma$ are diagonal in the same orthonormal basis $\{\ket{x}\}_{x \in \mathcal{X}}$, i.e.,
\begin{equation}
    \rho = \sum_{x \in \mathcal{X}} P(x) \ketbra{x}{x}, \quad \sigma = \sum_{x \in \mathcal{X}} Q(x) \ketbra{x}{x},
\end{equation}
then $\|\rho - \sigma\|_1 = \|P - Q\|_1$. Therefore, the trace distance generalizes the statistical distance to the quantum setting. Similarly to the statistical distance, the trace distance quantifies the distinguishability of states under arbitrary positive operator-valued measurement (POVM).

Here, we present several useful properties of the statistical distance and trace distance.

\begin{proposition}[Data processing inequality~\cite{Wilde_2013}]
    Let $\rho$ and $\sigma$ be two quantum states in a Hilbert space $\mathcal{H}$, and let $\mathcal{N}:\mathcal{H} \rightarrow \mathcal{H}'$ be a quantum channel. Then,
    \begin{equation}
        \|\mathcal{N}(\rho) - \mathcal{N}(\sigma)\|_1 \le \|\rho - \sigma\|_1.
    \end{equation}
    \label{prop:data_processing_inequality}
\end{proposition}

A useful consequence of the data processing inequality is that if $P$ and $Q$ are the output distributions of two $n$-qubit states $\rho$ and $\sigma$ under measurement in the computational basis, then the statistical distance between $P$ and $Q$ is upper bounded by the trace distance between $\rho$ and $\sigma$, i.e.,
\begin{equation}
    \|P - Q\|_1 \le \|\rho - \sigma\|_1.
\end{equation}
To see this, define $\rho'$ and $\sigma'$ as the dephased versions of $\rho$ and $\sigma$ in the computational basis, i.e., $\rho'= \sum_{x \in \{0,1\}^n}P(x)\ketbra{x}{x}$ and $\sigma' = \sum_{x \in \{0,1\}^n}Q(x)\ketbra{x}{x}$. Then, $\|P - Q\|_1 = \|\rho' - \sigma'\|_1$. Since $\rho'$ and $\sigma'$ are obtained from $\rho$ and $\sigma$ by applying complete dephasing channels on all qubits, the data processing inequality implies $\|\rho' - \sigma'\|_1 \le \|\rho - \sigma\|_1$.

The same reasoning also gives a data processing inequality for statistical distance: if $P$ and $Q$ are two probability distributions of a random variable $X$ taking values in $\mathcal{X}$, and $\mathcal{N}:\mathcal{X} \rightarrow \mathcal{Y}$ is a Markov kernel, then
\begin{equation}
     \|\mathcal{N}(P) - \mathcal{N}(Q)\|_1 \le \|P - Q\|_1.
\end{equation}
Note that conditional probabilities are special cases of Markov kernels, we will routinely use this property in our main analysis.

Along with the data processing inequality, we also remark the following useful property of the statistical distance.

\begin{proposition}\label{prop:reducing_marginals}
    Let $X$, $Y$, and $Z$ be random variables taking values in $\mathcal{X}$, $\mathcal{Y}$, and $\mathcal{Z}$, respectively. Let $R_{XY}$ be a joint distribution of $X$ and $Y$, and let $S_{Z|Y}$ and $T_{Z|Y}$ be Markov kernels of $Z$ given $Y$. Then
    \begin{equation}
        \|(S_{Z|Y} - T_{Z|Y})R_{XY}\|_1 
        = \|(S_{Z|Y} - T_{Z|Y})R_{Y}\|_1,
    \end{equation}
    where $R_Y$ is the marginal distribution of $R_{XY}$ on $Y$.
\end{proposition}

\begin{proof}
    \begin{equation}
    \begin{split}
        \left\|(S_{Z|Y} - T_{Z|Y})R_{XY}\right\|_1
        &= \sum_{x \in \mathcal{X}, y \in \mathcal{Y}, z \in \mathcal{Z}} R_{XY}(x,y) \left|S_{Z|Y}(z|y) - T_{Z|Y}(z|y)\right|\\
        &= \sum_{y \in \mathcal{Y}, z \in \mathcal{Z}} R_{Y}(y) \left|S_{Z|Y}(z|y) - T_{Z|Y}(z|y)\right|\\
        &= \left\|\left(S_{Z|Y} - T_{Z|Y}\right)R_{Y}\right\|_1,
    \end{split}
    \end{equation}
    since $R_Y(y) = \sum_{x \in \mathcal{X}} R_{XY}(x,y)$.
\end{proof}

Although this identity is simple to show, it is a crucial component of the proof of Lemma~\ref{lem:bounding_from_marginals}. Importantly, the quantum analogue does not hold in general. Specifically, for a bipartite state $\rho_{AB}$ and two quantum channels $\mathcal{N}_{B\rightarrow BC}$ and $\mathcal{M}_{B\rightarrow BC}$, one does not have
\begin{equation}
    \|\mathcal{N}_{B\rightarrow BC}(\rho_{AB}) - \mathcal{M}_{B\rightarrow BC}(\rho_{AB})\|_1
    = \|\mathcal{N}_{B\rightarrow BC}(\rho_B) - \mathcal{M}_{B\rightarrow BC}(\rho_B)\|_1,
\end{equation}
in general, where $\rho_B = \Tr_A(\rho_{AB})$ is the reduced density matrix on $B$. For example, take $\dim \mathcal{H}_A=\dim \mathcal{H}_B=\dim \mathcal{H}_C=2$, and let the input be the Bell pair
\begin{equation}
    \rho_{AB} = \frac{1}{2}\big(\ketbra{00}{00}+\ketbra{00}{11}+\ketbra{11}{00}+\ketbra{11}{11}\big).
\end{equation}
Define the channels
\begin{align}
    \mathcal{N}_{B\rightarrow BC}(\sigma) &= \sigma \otimes \ketbra{0}{0}_C,\\
    \mathcal{M}_{B\rightarrow BC}(\sigma) &= \frac{I_B}{2}{\rm Tr}(\sigma) \otimes \ketbra{0}{0}_C.
\end{align}
Since $\rho_B=I/2$, it follows that 
\begin{equation}
    \|\mathcal{N}_{B\rightarrow BC}(\rho_B) - \mathcal{M}_{B\rightarrow BC}(\rho_B)\|_1 = 0.
\end{equation}
However,
\begin{equation}
    \|\mathcal{N}_{B\rightarrow BC}(\rho_{AB}) - \mathcal{M}_{B\rightarrow BC}(\rho_{AB})\|_1 = \frac{3}{2}.
\end{equation}

\subsection{Effective shallow depth for Pauli observables}

To prove the effective shallow depth of noisy random circuits, we use the recent result in Ref.~\cite{mele2024noiseinducedshallowcircuitsabsence} which states that the expectation value of a Pauli observable of output states of a noisy random circuit with different input states converge to the same value.

To state this result, we first define the normal form of a single-qubit noise channel. A single-qubit quantum channel $\mathcal{N}$ can be succinctly described as a Pauli transfer matrix $\mathbf{T}(\mathcal{N})$, which is a $4\times 4$ matrix defined as
\begin{equation}
    \mathbf{T}(\mathcal{N}) =
    \begin{pmatrix}
        T_{II} & T_{IX} & T_{IY} & T_{IZ} \\
        T_{XI} & T_{XX} & T_{XY} & T_{XZ} \\
        T_{YI} & T_{YX} & T_{YY} & T_{YZ} \\
        T_{ZI} & T_{ZX} & T_{ZY} & T_{ZZ} \\
    \end{pmatrix},
\end{equation}
where $T_{PQ} = \Tr\left[P\mathcal{N}(Q)\right]/2$ for $P,Q \in \{I,X,Y,Z\}$. Given this representation, Ref.~\cite{KingRuskaiMinimalEntropyStates2001} showed that for any single-qubit quantum channel $\mathcal{N}$, there exists unitary matrices $U_1,U_2$ such that
\begin{equation}
    \mathbf{T}(\mathcal{U}_1 \circ \mathcal{N} \circ \mathcal{U}_2) =
    \begin{pmatrix}
        1 & 0 & 0 & 0 \\
        t_1 & \lambda_1 & 0 & 0 \\
        t_2 & 0 & \lambda_2 & 0 \\
        t_3 & 0 & 0 & \lambda_3 \\
    \end{pmatrix}
\end{equation}
for some real numbers $t_1,t_2,t_3$ and $\lambda_1,\lambda_2,\lambda_3 \in [0,1]$, where $\mathcal{U}_i(\cdot) = U_i(\cdot)U_i^\dagger$ for $i=1,2$. With this parametrization, Ref.~\cite{mele2024noiseinducedshallowcircuitsabsence} showed that the following:
\begin{proposition}
    For any single-qubit quantum channel $\mathcal{N}$, let $t_1,t_2,t_3$ and $\lambda_1,\lambda_2,\lambda_3$ be the parameters of the normal form of $\mathcal{N}$ as described above. Then,
    \begin{equation}
        \frac{1}{3}(t_1^2 + t_2^2 + t_3^2 + \lambda_1^2 + \lambda_2^2 + \lambda_3^2) \le 1,
    \end{equation}
    and the equality holds if and only if $\mathcal{N}$ is a unitary channel.
\end{proposition}

Since we only consider single-qubit noise channels that are not unitary, we always have $\frac{1}{3}(t_1^2 + t_2^2 + t_3^2 + \lambda_1^2 + \lambda_2^2 + \lambda_3^2) < 1$. With these normal forms, Ref.~\cite{mele2024noiseinducedshallowcircuitsabsence} proved that the expectation value of a Pauli observable of output states of a noisy random circuit with different input states converge exponentially to the same value. Specifically, denoting the set of all $n$-qubit Pauli operators as $\mathcal{P}_n = \{I,X,Y,Z\}^{\otimes n}$, we have the following proposition.

\begin{proposition}[Adapted from Ref.~\cite{mele2024noiseinducedshallowcircuitsabsence}]\label{prop:mele}
    Let $\rho$ and $\sigma$ be arbitrary density matrices of the qubits on $\Lambda$, and $\mathcal{C}$ be a depth-$d$ noisy random circuit. Let the noise channel $\mathcal{N}$ be a single-qubit quantum channel that is not unitary, and let $c = \frac{1}{3}(t_1^2 + t_2^2 + t_3^2 + \lambda_1^2 + \lambda_2^2 + \lambda_3^2)$ be the parameter of the normal form of $\mathcal{N}$ as described above. Denoting the output state of $\mathcal{C}(\rho)$ as $\rho'$ and the output state of $\mathcal{C}(\sigma)$ as $\sigma'$, we have
    \begin{equation}
        \mathbb{E}_{\mathcal{C}}\left[\Tr\left(P(\rho'-\sigma')\right)^{2}\right] \le 4c^{|P|+d-1},
    \end{equation}
    for all $n$-qubit Pauli operators $P \in \mathcal{P}_n$. Here, $|P|$ denotes the number of qubits that $P$ acts on nontrivially.
\end{proposition}

\subsection{Lightcone argument}

\begin{figure}
    \centering
    \includegraphics[width=\linewidth]{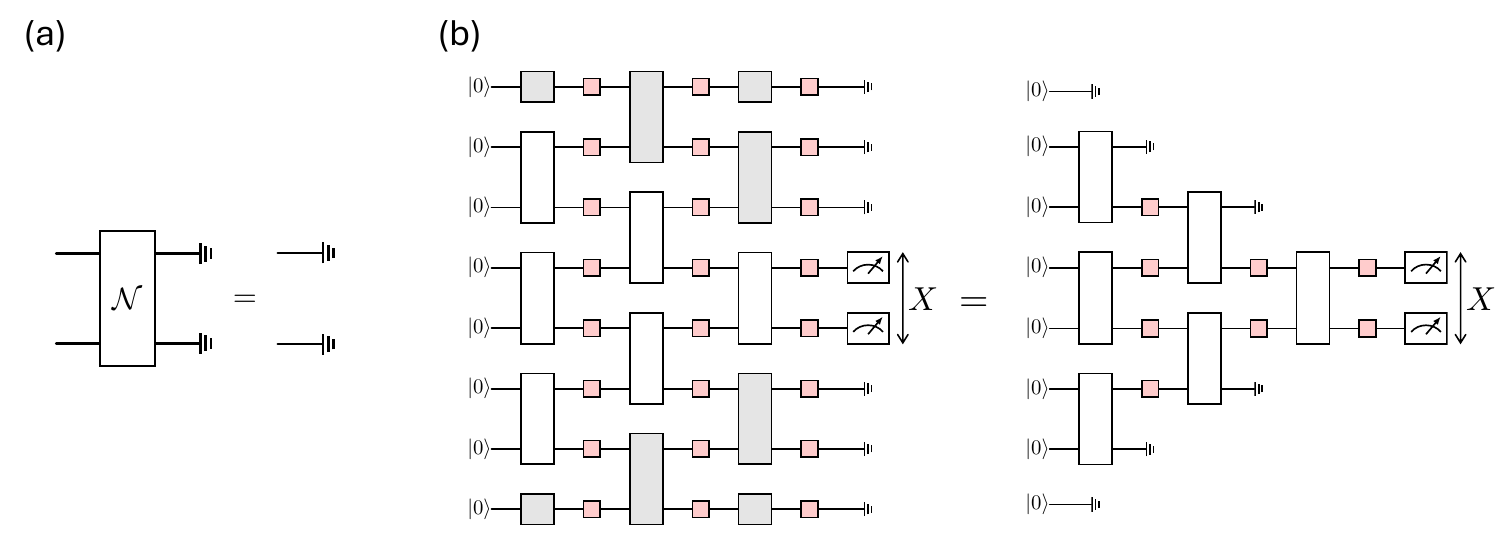}
    \caption{(a) The trace-preserving property of quantum channels. Here, the grounding symbol denotes applying partial trace. (b) The lightcone argument for computing marginal distributions. The white gates denote the ones inside the backward lightcone of $X$, and the gray gates are those outside the backward lightcone of $X$. As we trace out all qubits in $[n]\setminus X$, every gate outside the backward lightcone of $X$ vanishes.}
    \label{fig:lightcone}
\end{figure}

When analyzing the output distribution of a shallow-depth circuit, we use the lightcone argument to efficiently compute marginal distributions. Specifically, we can compute the marginal distribution of a subset of qubits $X$ by tracing out all qubits outside the lightcone of $X$ and obtaining the reduced density matrix $\rho_X$. Let $\mathcal{N}$ be an arbitrary quantum channel that acts only on $[n]\setminus X$. Writing the channel in Kraus form as $\mathcal{N}(\cdot) = \sum_k E_k (\cdot) E_k^\dagger$, we have
\begin{align}
    \Tr_{[n]\setminus X}[\mathcal{N}(\rho)]
    &= \Tr_{[n]\setminus X}\left[\sum_k (E_k \otimes I_X) \rho (E_k^\dagger \otimes I_{[n]\setminus X})\right]\\
    &= \Tr_{[n]\setminus X}\left[\sum_k (E_k^\dagger E_k \otimes I_X)\rho\right]\\
    &= \rho_X,
\end{align}
because of the trace-preserving property of quantum channels, $\sum_k E_k^\dagger E_k = I$ [Fig.~\ref{fig:lightcone}(a)]. Therefore, any quantum channel that acts only outside of $X$ does not affect $\rho_X$. This directly implies that the marginal distribution of $X$ can be computed only by the channels that has causal influence on $X$: By tracing out all qubits outside of $X$, all qubits and the channels inside the backward lightcone of $X$ vanish by repeatedly applying the above relation [Fig.~\ref{fig:lightcone}(b)]. Therefore, we can compute the marginal distribution of $X$ by simulating the quantum circuit on the qubits inside the backward lightcone of $X$.

\section{\label{sec:deferred_proofs}Deferred proofs}

In this section, we provide the deferred proofs of the main theorem in the main text. We begin with the proof of Proposition~\ref{prop:patches-restated}, which is a central component of the proofs of both Theorem~\ref{thm:effective_shallow_depth} and \ref{thm:shallow-depth_algorithm}.

\begingroup
  \renewcommand{\theproposition}{\ref{prop:patches}}
    \begin{proposition}[restated]
        Let $P$ be a probability distribution that satisfies the $\eta(\ell)$-approximate Markov condition. Then
        \begin{equation}
            \left\| P - \prod_{i=1}^{|V|} P_{X_i|N'(X_i)} \right\|_1 \le O(n/\ell^D) \sqrt{\eta(\ell)},\label{eq:prop1-eq1}
        \end{equation}
        where $N'(X_i) = N(X_i) \cap X_{<i}$. In addition, if $P$ is the output distribution of a noisy random circuit $\mathcal{C}$ that satisfies the average $\eta(\ell)$-approximate Markov condition, then
        \begin{equation}
            \mathbb{E}_{\mathcal{C}} \left\| P - \prod_{i=1}^{|V|} P_{X_i|N'(X_i)} \right\|_1 \le O(n/\ell^D) \sqrt{\eta(\ell)}.\label{eq:prop1-eq2}
        \end{equation}
        \label{prop:patches-restated}
    \end{proposition}
\endgroup

\begin{proof}
    We first show that the first inequality in Eq.~\eqref{eq:prop1-eq1}. We begin by showing the following relation:
    \begin{equation}
        \left\| P_{X_{<(j+1)}} - \prod_{i=1}^{j} P_{X_i|N'(X_i)} \right\|_1 \le \sqrt{2\ln2 \cdot \eta(\ell)} + \left\| P_{X_{<j}} - \prod_{i=1}^{j-1} P_{X_i|N'(X_i)} \right\|_1
        \label{eq:prop1-telescoping-step}
    \end{equation}
    for all $j=2,\dots, |V|$, where $P_{X_{<j}}$ denotes the marginal distribution of $P$ on $X_{<j} = X_1 \bigcup_{i=1}^{j-1} X_{i}$ and $P_{<(|V|+1)}=P$. To this end, note that $P_{X_{<(j+1)}} = P_{X_{<j}}P_{X_j|X_{<j}}$ by chain rule. Then,
    \begin{equation}
    \begin{split}
        \left\|P_{X_{<(j+1)}} - \prod_{i=1}^{j} P_{X_i|N'(X_i)}\right\|_1
        &= \left\|P_{X_{<j}}P_{X_j|X_{<j}} - \left(\prod_{i=1}^{j-1} P_{X_i|N'(X_i)}\right)P_{X_j|N'(X_j)}\right\|_1\\
        &\le \left\|P_{X_{<j}}P_{X_j|X_{<j}} - P_{X_{<j}}P_{X_j|N'(X_j)}\right\|_1
        + \left\|
            P_{X_{<j}}P_{X_{j}|N'(X_{j})} - \prod_{i=1}^{j-1} P_{X_i|N'(X_i)}P_{X_{j}|N'(X_{j})}
        \right\|_1\\
        &\le \sqrt{2\ln2I\left(X_{j}:\left(X_{<j}\setminus N'(X_{j})\right)|N'(X_{j})\right)} + \left\|
            \left(P_{X_{<j}} - \prod_{i=1}^{j-1} P_{X_i|N'(X_i)}\right)P_{X_{j}|N'(X_{j})}
        \right\|_1\\
        &\le \sqrt{2\ln2I\left(X_{j}:\left(X_{<j}\setminus N'(X_{j})\right)|N'(X_{j})\right)} + \left\|
            P_{X_{<j}} - \prod_{i=1}^{j-1} P_{X_i|N'(X_i)}
        \right\|_1
    \end{split}
    \end{equation}
    Here, the first inequality follows from the triangle inequality, the second inequality is from the Pinsker's inequality in Eq.~\eqref{eq:pinsker}, and we use data processing inequality in the last inequality. Since ${\rm dist}(X_{j}, X_{<j}\setminus N'(X_{j}))\ge \ell$, $\eta(\ell)$-approximate Markov condition leads to Eq.~\eqref{eq:prop1-telescoping-step}.

    Now, we can apply the telescoping sum to Eq.~\eqref{eq:prop1-telescoping-step} for $j=2,\dots, |V|$, which gives
    \begin{equation}
        \left\| P - \prod_{i=1}^{|V|} P_{X_i|N'(X_i)} \right\|_1 \le (|V|-1) \sqrt{2\ln2 \cdot \eta(\ell)} + \left\| P_{X_{1}} - P_{X_1|N'(X_1)} \right\|_1,
    \end{equation}
    where the second term is zero since $N'(X_1)=\emptyset$. Since $|V| = O(n/\ell^D)$, we obtain the first inequality Eq.~\eqref{eq:prop1-eq1}.

    The second inequality Eq.~\eqref{eq:prop1-eq2} can be proved in a similar way with a small modification. With the same procedure as above while taking the expectation over the noisy random circuit $\mathcal{C}$, we have
    \begin{equation}
        \mathbb{E}_{\mathcal{C}}\left\| P_{X_{<(j+1)}} - \prod_{i=1}^{j} P_{X_i|N'(X_i)} \right\|_1 \le \mathbb{E}_{\mathcal{C}}\left[\sqrt{2\ln2I\left(X_{j}:\left(X_{<j}\setminus N'(X_{j})\right)|N'(X_{j})\right)}\right] + \mathbb{E}_{\mathcal{C}}\left\| P_{X_{<j}} - \prod_{i=1}^{j-1} P_{X_i|N'(X_i)} \right\|_1.
    \end{equation}
    By Jensen's inequality, we have
    \begin{equation}
        \mathbb{E}_{\mathcal{C}}\left[\sqrt{2\ln2I\left(X_{j}:\left(X_{<j}\setminus N'(X_{j})\right)|N'(X_{j})\right)}\right] \le \sqrt{2\ln2 \cdot \mathbb{E}_{\mathcal{C}}\left[I\left(X_{j}:\left(X_{<j}\setminus N'(X_{j})\right)|N'(X_{j})\right)\right]},
    \end{equation}
    which is bounded by $\sqrt{2\ln2 \cdot \eta(\ell)}$ by the $\eta(\ell)$-approximate Markov condition. Therefore, we have the same relation as Eq.~\eqref{eq:prop1-telescoping-step} for the expectation value, and it leads to Eq.~\eqref{eq:prop1-eq2}.
\end{proof}

\subsection{\label{sec:proof-theorem1}Proof of Theorem~\ref{thm:effective_shallow_depth}}

We now prove Theorem~\ref{thm:effective_shallow_depth}. To this end, we introduce two lemmas that will be used in the proof. Given two arbitrary density matrices $\rho$ and $\sigma$, let $P$ and $Q$ be the output distributions of a noisy random circuit $\mathcal{C}$ with input states $\rho$ and $\sigma$, respectively. Lemma~\ref{lem:bounding_from_marginals} shows that the distance between $P$ and $Q$ can be bounded by the distances between their marginal distributions on small regions, while Lemma~\ref{lem:indistinguishable_marginals} shows that those distances between marginal distributions of $P$ and $Q$ quickly become indistinguishable.

\begin{lemma}
    For a depth-$d$ noisy random circuit $\mathcal{C}$, let $P$ and $Q$ be the output distributions of $\mathcal{C}(\rho)$ and $\mathcal{C}(\sigma)$, respectively. If both $P$ and $Q$ satisfy the $\eta(\ell)$-approximate Markov condition, we have
    \begin{equation}
        \mathbb{E}_{\mathcal{C}} \|P - Q\|_{1} 
        \le \sum_{i=1}^{|V|} \left(
            \mathbb{E}_{\mathcal{C}}\left\|P_{X_i \sqcup N'(X_i)} - Q_{X_i \sqcup N'(X_i)}\right\|_1
            + \mathbb{E}_{\mathcal{C}}\left\|P_{N'(X_i)} - Q_{N'(X_i)}\right\|_1
        \right)
        + O(n / \ell^D) \sqrt{\eta(\ell)}.
    \label{eq:bounding_from_marginals}
    \end{equation}
    \label{lem:bounding_from_marginals}
\end{lemma}

\begin{proof}
    We first show the following recursive relation: for $j=2,\dots, |V|$,
    \begin{equation}
        \mathbb{E}_{\mathcal{C}}\|P_{X_{<(j+1)}} - Q_{X_{<(j+1)}}\|_1
        \le 2\sqrt{\eta(\ell)}
        + \mathbb{E}_{\mathcal{C}}\left\|P_{X_j \sqcup N'(X_j)} - Q_{X_j \sqcup N'(X_j)}\right\|_1
        + \mathbb{E}_{\mathcal{C}}\left\|P_{N'(X_j)} - Q_{N'(X_j)}\right\|_1.
        \label{eq:lem1_step1}
    \end{equation}
    To this end, note that $P_{X_{<(j+1)}} = P_{X_{<j}}P_{X_j|X_{<j}}$ and $Q_{X_{<(j+1)}} = Q_{X_{<j}}Q_{X_j|X_{<j}}$ by the chain rule. Then, we have
    \begin{equation}
    \begin{split}
        \mathbb{E}_{\mathcal{C}}\|P_{X_{<(j+1)}} - Q_{X_{<(j+1)}}\|_1
        &= \mathbb{E}_{\mathcal{C}}\left\|P_{X_{<j}}P_{X_j|X_{<j}} - Q_{X_{<j}}Q_{X_j|X_{<j}}\right\|_1\\
        &\le \mathbb{E}_{\mathcal{C}}\left\|P_{X_{<j}}P_{X_{j}|X_{<j}} - P_{X_{<j}}P_{X_{j}|N'(X_{j})}\right\|_1 + \mathbb{E}_{\mathcal{C}}\left\| P_{X_{<j}}P_{X_{j}|N'(X_{j})} - Q_{X_{<j}}Q_{X_{j}|N'(X_{j})} \right\|_1 \\
        & \quad + \mathbb{E}_{\mathcal{C}}\left\|Q_{X_{<j}}Q_{X_{j}|N'(X_{j})} - Q_{X_{<j}}Q_{X_{j}|X_{<j}}\right\|_1\\
        &\le \mathbb{E}_{\mathcal{C}}\left\| P_{X_{<j}}P_{X_{j}|N'(X_{j})} - Q_{X_{<j}}Q_{X_{j}|N'(X_{j})} \right\|_1\\
        &\qquad + \mathbb{E}_{\mathcal{C}}\sqrt{2\ln2 I\left(X_{j}:\left(X_{<j}\setminus N'(X_{j})\right)|N'(X_{j})\right)_P}\\
        &\qquad\qquad + \mathbb{E}_{\mathcal{C}}\sqrt{2\ln2 I\left(X_{j}:\left(X_{<j}\setminus N'(X_{j})\right)|N'(X_{j})\right)_Q},
    \end{split}
    \end{equation}
    where the first inequality follows from the triangle inequality, the second inequality is from the Pinsker's inequality in Eq.~\eqref{eq:pinsker}. Here, we denote $I(X:Y|Z)_P$ and $I(X:Y|Z)_Q$ for CMIs for the distributions $P$ and $Q$, respectively. By Jensen's inequality, we have
    \begin{align}
        \mathbb{E}_{\mathcal{C}}\sqrt{2\ln2 I\left(X_{j}:\left(X_{<j}\setminus N'(X_{j})\right)|N'(X_{j})\right)_P}
        &\le \sqrt{2\ln2 \cdot \mathbb{E}_{\mathcal{C}}I\left(X_{j}:\left(X_{<j}\setminus N'(X_{j})\right)|N'(X_{j})\right)_P},\\
        \mathbb{E}_{\mathcal{C}}\sqrt{2\ln2 I\left(X_{j}:\left(X_{<j}\setminus N'(X_{j})\right)|N'(X_{j})\right)_Q}
        &\le \sqrt{2\ln2 \cdot \mathbb{E}_{\mathcal{C}}I\left(X_{j}:\left(X_{<j}\setminus N'(X_{j})\right)|N'(X_{j})\right)_Q},
    \end{align}
    which are both bounded by $\sqrt{2\ln2 \cdot \eta(\ell)}$ by the $\eta(\ell)$-approximate Markov condition. Therefore, we have
    \begin{equation}\label{eq:lem1_step1-1}
        \mathbb{E}_{\mathcal{C}}\|P_{X_{<(j+1)}} - Q_{X_{<(j+1)}}\|_1
        \le 2\sqrt{2\ln 2 \cdot \eta(\ell)} + \mathbb{E}_{\mathcal{C}}\left\| P_{X_{<j}}P_{X_{j}|N'(X_{j})} - Q_{X_{<j}}Q_{X_{j}|N'(X_{j})} \right\|_1
    \end{equation}
    We can further bound the second term as follows:
    \begin{equation}
    \begin{split}
        \mathbb{E}_{\mathcal{C}}\left\| P_{X_{<j}}P_{X_{j}|N'(X_{j})} - Q_{X_{<j}}Q_{X_{j}|N'(X_{j})} \right\|_1
        &\le \mathbb{E}_{\mathcal{C}}\left\| P_{X_{<j}}P_{X_{j}|N'(X_{j})} - P_{X_{<j}}Q_{X_{j}|N'(X_{j})} \right\|_1\\
        &\qquad + \mathbb{E}_{\mathcal{C}}\left\| P_{X_{<j}}Q_{X_{j}|N'(X_{j})} - Q_{X_{<j}}Q_{X_{j}|N'(X_{j})} \right\|_1\\
        &= \mathbb{E}_{\mathcal{C}}\left\| P_{X_{<j}}\left(P_{X_{j}|N'(X_{j})} - Q_{X_{j}|N'(X_{j})}\right) \right\|_1\\
        &\qquad + \mathbb{E}_{\mathcal{C}}\left\| \left(P_{X_{<j}} - Q_{X_{<j}}\right)Q_{X_{j}|N'(X_{j})} \right\|_1 \\
        &\le \mathbb{E}_{\mathcal{C}}\left\| P_{N'(X_{j})} \left(P_{X_{j}|N'(X_{j})} - Q_{X_{j}|N'(X_{j})}\right) \right\|_1 + \mathbb{E}_{\mathcal{C}}\left\|P_{X_{<j}} - Q_{X_{<j}} \right\|_1,
    \end{split}
    \end{equation}
    where the first inequality follows from the triangle inequality. For the second inequality, we use Proposition~\ref{prop:reducing_marginals} for the first term, and the data processing inequality for the second term. Using the triangle inequality and the data processing inequality once again, we have
    \begin{equation}
    \begin{split}
        \mathbb{E}_{\mathcal{C}}\left\| P_{X_{<j}}P_{X_{j}|N'(X_{j})} - Q_{X_{<j}}Q_{X_{j}|N'(X_{j})} \right\|_1
        &\le \mathbb{E}_{\mathcal{C}}\left\| P_{N'(X_{j})} \left(P_{X_{j}|N'(X_{j})} - Q_{X_{j}|N'(X_{j})}\right) \right\|_1 + \mathbb{E}_{\mathcal{C}}\left\|P_{X_{<j}} - Q_{X_{<j}} \right\|_1\\
        &\le \mathbb{E}_{\mathcal{C}}\left\| P_{N'(X_{j})} P_{X_{j}|N'(X_{j})} - Q_{N'(X_{j})}Q_{X_{j}|N'(X_{j})} \right\|_1\\
        &\qquad + \mathbb{E}_{\mathcal{C}}\left\|\left(Q_{N'(X_{j})} - P_{N'(X_{j})}\right)Q_{X_{j}|N'(X_{j})} \right\|_1 + \mathbb{E}_{\mathcal{C}}\left\|P_{X_{<j}} - Q_{X_{<j}} \right\|_1\\
        &\le \mathbb{E}_{\mathcal{C}}\left\| P_{X_{j} \sqcup N'(X_{j})} - Q_{X_{j} \sqcup N'(X_{j})} \right\|_1\\
        &\qquad + \mathbb{E}_{\mathcal{C}}\left\|Q_{N'(X_{j})} - P_{N'(X_{j})} \right\|_1 + \mathbb{E}_{\mathcal{C}}\left\|P_{X_{<j}} - Q_{X_{<j}} \right\|_1
    \end{split}
    \end{equation}
    Plugging this into Eq.~\eqref{eq:lem1_step1-1}, we have the desired recursion relation of Eq.~\eqref{eq:lem1_step1}.

    By applying the telescoping sum to Eq.~\eqref{eq:lem1_step1} for $j=2,\dots, |V|$, we have
    \begin{equation}
    \begin{split}
        \|P-Q\|_1 &\le \left\| P_{X_1} - Q_{X_1} \right\|_1\\
        &\qquad + \sum_{i=2}^{|V|} \left(\left\| P_{X_i \sqcup N'(X_i)} - Q_{X_i \sqcup N'(X_i)}\right\|_1 + \left\|P_{N'(X_i)} - Q_{N'(X_i)}\right\|_1\right) + O(n/\ell^D) \sqrt{\eta(\ell)}.
    \end{split}
    \end{equation}
    Since $N'(X_1) = \emptyset$, we have $\left\| P_{X_1} - Q_{X_1} \right\|_1 = \left\| P_{X_1\sqcup N'(X_1)} - Q_{X_1\sqcup N'(X_1)} \right\|_1$, which concludes the proof.
\end{proof}

\begin{lemma}
    For a depth-$d$ noisy random circuit $\mathcal{C}$, let $P$ and $Q$ be the output distributions of $\mathcal{C}(\rho)$ and $\mathcal{C}(\sigma)$, respectively, for arbitrary input states $\rho$ and $\sigma$. Then, the marginal distributions of $P$ and $Q$ on $X \subset \Lambda$, denoted by $P_X$ and $Q_X$, satisfy
    \begin{equation}
        \mathbb{E}_{\mathcal{C}}\|P_X - Q_X\|_1 \le 2^{|X|} \exp(-\Omega(d)).
    \end{equation}
    \label{lem:indistinguishable_marginals}
\end{lemma}

\begin{proof}
    By Jensen's inequality, we have
    \begin{equation}
        \left(\mathbb{E}_{\mathcal{C}}\|\mathcal{C}(\rho)_A - \mathcal{C}(\sigma)_A\|_1\right)^{2}
        \le
        \mathbb{E}_{\mathcal{C}}\|\mathcal{C}(\rho)_A - \mathcal{C}(\sigma)_A\|_{1}^{2}.
    \end{equation}
    Meanwhile, by Cauchy-Schwarz inequality, we have
    \begin{equation}
    \begin{split}
        \left\|\mathcal{C}(\rho)_A - \mathcal{C}(\sigma)_A\right\|_{1}^{2}
        &\le 2^{|A|}\left\|\mathcal{C}(\rho)_A - \mathcal{C}(\sigma)_A\right\|_{2}^{2} \\
        &= \sum_{P\in\mathcal{P}_{|A|}} \Tr(P(\mathcal{C}(\rho)_A - \mathcal{C}(\sigma)_A))^{2} \\
        &= \sum_{P'\in\mathcal{P}_n: P_i = I \, \forall i \in [n]\setminus A}{\rm Tr}(P'(\mathcal{C}(\rho) - \mathcal{C}(\sigma)))^{2},
    \end{split}
    \end{equation}
    where $P_i$ denotes the Pauli matrix that $P$ acts on the $i$-th qubit. By Proposition~\ref{prop:mele}, we have
    \begin{equation}
    \begin{split}
        \left(\mathbb{E}_{\mathcal{C}}\left\|\mathcal{C}(\rho)_A - \mathcal{C}(\sigma)_A\right\|_{1}\right)^{2}
        &\le \sum_{P'\in\mathcal{P}_n: P_i = I \, \forall i \in [n]\setminus A} 4c^{|P|+d-1} \\
        &\le 4^{1+|A|} c^{d-1},
    \end{split}
    \end{equation}
    for some parameter $0<c<1$. Finally, by the data processing inequality, we have
    \begin{equation}
    \begin{split}
        \mathbb{E}_{\mathcal{C}}\left\|P_A - Q_A\right\|_1
        &\le \mathbb{E}_{\mathcal{C}}\left\|\mathcal{C}(\rho)_A - \mathcal{C}(\sigma)_A\right\|_{1}\\
        &\le 2^{|A|+1}c^{(d-1)/2}.
    \end{split}
    \end{equation}
\end{proof}

With Lemmas~\ref{lem:bounding_from_marginals} and \ref{lem:indistinguishable_marginals}, we provide the proof of Theorem~\ref{thm:effective_shallow_depth}.

\begingroup
  \renewcommand{\thetheorem}{\ref{thm:effective_shallow_depth}}
    \begin{theorem}[restated]
        Let $\rho$ and $\sigma$ be arbitrary density matrices over the qubits in a $D$-dimensional grid $\Lambda$, and let $\mathcal{C}$ be a depth-$d$ noisy random circuit. Denote the output distributions of $\mathcal{C}(\rho)$ and $\mathcal{C}(\sigma)$ by $P$ and $Q$, respectively. Suppose both distributions satisfy the average $\mathrm{poly}(n)\exp(-\Omega(\ell))$-approximate Markov condition for all $\ell$. Then, for any $\varepsilon > 0$, there exists $d^* = O(\log^D(n/\varepsilon))$ such that
        \begin{equation}
            \mathbb{E}_{\mathcal{C}}\|P-Q\|_1  \le \varepsilon,
        \end{equation}
        for all $d \ge d^*$.
        \label{thm:effective_shallow_depth-restated}
    \end{theorem}
\endgroup

\begin{proof}
    For each partition $X \in V$, we have $|X| = \ell^D$ and $|N'(X)| \le (3^D-1)\ell^{D}$. Therefore, by Lemma~\ref{lem:indistinguishable_marginals}, we have
    \begin{equation}
        \mathbb{E}_{\mathcal{C}}\left(\|P_{X\sqcup N'(X)} - Q_{X \sqcup  N'(X)}\|_1 + \|P_{N'(X)} - Q_{N'(X)}\|_1\right) \le 2^{O(\ell^D)}\exp(-\Omega(d)).    
    \end{equation}
    Plugging this in Lemma~\ref{lem:bounding_from_marginals}, we obtain
    \begin{equation}
        \mathbb{E}_{\mathcal{C}}\|P-Q\|_1 \le \frac{n}{\ell^D} \cdot 2^{O(\ell^D)}\exp(-\Omega(d)) + {\rm poly}(n) \exp(-\Omega(\ell)).
    \end{equation}
    By the choice of $\ell = O(\log(n/\varepsilon))$, we have the second term less than $\varepsilon/2$. Finally, we can choose $d^* = O(\ell^D)$ such that the first term is also less than $\varepsilon/2$, which concludes the proof.
\end{proof}


\subsection{\label{sec:proof-theorem3}Proof of Theorem~\ref{thm:main_theorem}}

Finally, we combine Theorems~\ref{thm:effective_shallow_depth} and~\ref{thm:shallow-depth_algorithm} to establish our main result, Theorem~\ref{thm:main_theorem}.

\begingroup
  \renewcommand{\thetheorem}{\ref{thm:main_theorem}}
    \begin{theorem}[restated]
        Let $\mathcal{C}$ be a depth-$d$ noisy random circuit on a $D$-dimensional grid $\Lambda$, and $P$ be the output distribution of $\mathcal{C}(\ketbra{0^n}{0^n})$. Suppose $P$ satisfies the average ${\rm poly}(n)\exp(-\Omega(\ell))$-approximate Markov condition for all $\ell$ and $d$. Then there exists a classical algorithm that outputs a sample from $P'$ such that $\|P-P'\|_1 \le \varepsilon$ with probability at least $1-\delta$ over the choice of $\mathcal{C}$, in runtime ${\rm poly}(n, 1/\varepsilon, 1/\delta)$ for $D=1$ and ${\rm quasipoly}(n, 1/\varepsilon, 1/\delta)$ for $D \ge 2$.
        \label{thm:main_theorem-restated}
    \end{theorem}
\endgroup

\begin{proof}
    By Theorem~\ref{thm:effective_shallow_depth}, we can choose $d^* = O(\log^D(n/\delta\varepsilon))$ such that a random circuit $\mathcal{C}'$ with depth greater than or equal to $d^*$ satisfies $\mathbb{E}_{\mathcal{C}'}\|Q-R\|_1 \le \delta\varepsilon/2$,
    where $Q$ and $R$ are the output distributions of $\mathcal{C}'$ with arbitrary input states, respectively. Then, for $d' = \min\{d, d^*\}$, we divide the given circuit $\mathcal{C}$ into two parts: $\mathcal{C} = \mathcal{C}_{2} \circ \mathcal{C}_{1}$, where $\mathcal{C}_{1}$ consists of the first $d - d'$ layers and $\mathcal{C}_{2}$ consists of the final $d'$ layers (here, note that if $d \le d^*$, we have $\mathcal{C}=\mathcal{C}_2$). Therefore, denoting $P'$ as the output distribution of $\mathcal{C}_{2}(\ketbra{0^n}{0^n})$, we have
    \begin{equation}
        \mathbb{E}_{\mathcal{C}_2}\|P - P'\| \le \delta\varepsilon/2.
    \end{equation}

    We now introduce another distribution $P''$ defined as
    \begin{equation}
        P'' = \prod_{i=1}^{|V|} P'_{X_i|N'(X_i)},
    \end{equation}
    with the coarse-graining with respect to $\ell$. Then, by Proposition~\ref{prop:patches-restated}, we can choose $\ell = O(\log(n/\delta\varepsilon))$ such that
    \begin{equation}
        \mathbb{E}_{\mathcal{C}}\|P' - P''\|_1 \le \delta\varepsilon/2.
    \end{equation}
    By the triangle inequality, we have
    \begin{equation}
        \mathbb{E}_{\mathcal{C}}\|P - P''\|_1 \le \delta\varepsilon,
    \end{equation}
    and Markov's inequality gives
    \begin{equation}
        \mathbb{P}\left[\|P-P''\|_1 \ge \varepsilon\right] \le \delta.
    \end{equation}
    In other words, with probability at least $1-\delta$, we have $\|P-P''\|_1 \le \varepsilon$.

    Finally, we can sample from $P''$ by the algorithm in Theorem~\ref{thm:shallow-depth_algorithm}, and the runtime of the algorithm is
    \begin{equation}
        \frac{n}{\ell^D} \cdot \exp\left(O(d'(\ell + 2d')^{D-1})\right) = \exp\left(O\left(\log^{D^2}(n/\delta\varepsilon)\right)\right),
    \end{equation}
    which is ${\rm poly}(n, 1/\varepsilon, 1/\delta)$ for $D=1$ and ${\rm quasipoly}(n, 1/\varepsilon, 1/\delta)$ for $D \ge 2$.
\end{proof}

\section{\label{sec:extended_results}Extended results beyond the grid geometry}

In this section, we present an algorithm that is applicable beyond the grid geometry. To this end, we consider a general graph $G=(V,E)$, where each vertex in $V$ corresponds to a qubit and edges represent the possible locations of gates between two qubits. Rather than using coarse-graining method and sampling partition-by-partition as in the main text, it is more convenient to sample it bit-by-bit. To this end, we prove the analog of Proposition~\ref{prop:patches-restated} for a general graph $G$. Enumerating each qubit with an arbitrary order, $V=\{v_1,\dots,v_n\}$, and denoting a ball with radius $\ell$ centered at $v_i$ as $B_\ell(v_i) = \{v_j \in V: d(v_i, v_j) \le \ell\}$, we have the following proposition.

\begin{proposition}
    Let $P$ be a probability distribution over $V$ that satisfies the $\eta(\ell)$-approximate Markov condition. Then
    \begin{equation}
        \left\| P - \prod_{i=1}^{n} P_{v_i|B'_\ell(v_i)} \right\|_1 \le O(n) \sqrt{\eta(\ell)},
    \end{equation}
    where $B'_\ell(v_i) = B_\ell(v_i) \cap v_{<i}$. In addition, if $P$ is the output distribution of a noisy random circuit $\mathcal{C}$ that satisfies the average $\eta(\ell)$-approximate Markov condition, then
    \begin{equation}
        \mathbb{E}_{\mathcal{C}} \left\| P - \prod_{i=1}^{n} P_{v_i|B'_\ell(v_i)} \right\|_1 \le O(n) \sqrt{\eta(\ell)}.
    \end{equation}
    \label{prop:patches-general}
\end{proposition}

\begin{proof}
    We begin by showing the following relation:
    \begin{equation}
        \left\| P_{v_{<(j+1)}} - \prod_{i=1}^{j} P_{v_i|B'_\ell(v_i)} \right\|_1 \le \sqrt{2\ln2 \cdot \eta(\ell)} + \left\| P_{v_{<j}} - \prod_{i=1}^{j-1} P_{v_i|B'_\ell(v_i)} \right\|_1
        \label{eq:prop7-telescoping-step}
    \end{equation}
    for all $j=2,\dots, n$. Then, as in Proposition~\ref{prop:patches-restated}, we can apply the triangle inequality, Pinsker's inequality, data processing inequality, and Markov condition to obtain it:
    \begin{equation}
    \begin{split}
        \left\|P_{v_{<(j+1)}} - \prod_{i=1}^{j} P_{v_i|B'_\ell(v_i)}\right\|_1
        &= \left\|P_{v_{<j}}P_{v_j|v_{<j}} - \left(\prod_{i=1}^{j-1} P_{v_i|B'_\ell(v_i)}\right)P_{v_j|B'_\ell(v_j)}\right\|_1\\
        &\le \left\|P_{v_{<j}}P_{v_{j}|v_{<j}} - P_{v_{<j}}P_{v_{j}|B'_\ell(v_{j})}\right\|_1
        + \left\|
            P_{X_{<j}}P_{X_{j}|B'_\ell(v_{j})} - \prod_{i=1}^{j-1} P_{X_i|N'(X_i)}P_{X_{j}|N'(X_{j})}
        \right\|_1\\
        &\le \sqrt{2\ln2I\left(v_{j}:\left(v_{<j}\setminus B'_\ell(v_{j})\right)|B'_\ell(v_{j})\right)} + \left\|
            \left(P_{v_{<j}} - \prod_{i=1}^{j-1} P_{v_i|B'_\ell(v_i)}\right)P_{v_{j}|B'_\ell(v_{j})}
        \right\|_1\\
        &\le \sqrt{2\ln2 \cdot \eta(\ell)} + \left\|
            P_{v_{<j}} - \prod_{i=1}^{j-1} P_{v_i|B'_\ell(v_i)}
        \right\|_1
    \end{split}
    \end{equation}
    Then we apply the telescoping sum to Eq.~\eqref{eq:prop7-telescoping-step} for $j=2,\dots, n$, which gives the first inequality in the proposition.

    The second inequality can be proved in a similar way with taking the expectation over the noisy random circuit $\mathcal{C}$ and Jensen's inequality, as in Proposition~\ref{prop:patches-restated}.
\end{proof}

We also have the analog of Theorem~\ref{thm:effective_shallow_depth} for a random circuit on a general graph $G$.

\begin{theorem}
    Let $\mathcal{C}$ be a depth-$d$ noisy random circuit on a graph $G=(V,E)$, and $P$ be the output distribution of $\mathcal{C}(\ketbra{0^n}{0^n})$. Suppose $P$ satisfies the average ${\rm poly}(n)\exp(-\Omega(\ell))$-approximate Markov condition and $\max_{v \in V}|B_\ell(v)|\le f(\ell)$ for all $\ell$, for some function $f$. Then, for any $\varepsilon > 0$, there exists $d^* = O(f(\log(n/\varepsilon)))$ such that
    \begin{equation}
        \mathbb{E}_{\mathcal{C}}\|P-Q\|_1  \le \varepsilon,
    \end{equation}
    for all $d \ge d^*$.
    \label{thm:effective_shallow_depth-general}
\end{theorem}

\begin{proof}
    First, we use the same procedure of the proof of Lemma~\ref{lem:bounding_from_marginals} to obtain
    \begin{equation}
        \mathbb{E}_{\mathcal{C}} \|P - Q\|_{1} 
        \le \sum_{i=1}^{n} \left(
            \mathbb{E}_{\mathcal{C}}\left\|P_{\{v_i\} \sqcup B'_\ell(v_i)} - Q_{\{v_i\} \sqcup B'_\ell(v_i)}\right\|_1
            + \mathbb{E}_{\mathcal{C}}\left\|P_{B'_\ell(v_i)} - Q_{B'_\ell(v_i)}\right\|_1
        \right)
        + O(n)\sqrt{\eta(\ell)}.
    \label{eq:bounding_from_marginals-general}
    \end{equation}
    for all $d \ge d^*$. Combining this with Lemma~\ref{lem:indistinguishable_marginals}, we have
    \begin{equation}
    \begin{split}
        \mathbb{E}_{\mathcal{C}} \|P - Q\|_{1} 
        &\le \sum_{i=1}^{n} \left(
            3 \cdot 2^{|B'_\ell(v_i)|} \exp(-\Omega(d))
        \right)
        + O(n)\sqrt{\eta(\ell)}\\
        &\le O(n) \cdot 2^{f(\ell)} \exp(-\Omega(d)) + O(n)\sqrt{\eta(\ell)}.
    \end{split}
    \end{equation}
    The second term becomes less than $\varepsilon/2$ by the choice of $\ell = O(\log(n/\varepsilon))$, and we can choose $d^* = O(f(\log(n/\varepsilon)))$ such that the first term is also less than $\varepsilon/2$, which concludes the proof.
\end{proof}

Therefore, the output distribution of a depth-$d$ noisy random circuit on a general graph $G$ can be approximated by the one of a depth-$d^*$ circuit with $d^* = O(\log(f(n/\varepsilon)))$. To establish the sampling algorithm for a general graph $G$, we introduce the local tree-width of a graph.

\begin{definition}[Local tree-width]
    Given a graph $G=(V,E)$, let $G[W]$ be the subgraph induced by a subset $W \subset V$. Then, the local tree-width of $G$ is defined as
    \begin{equation}
        \mathbf{ltw}(G, \ell) = \max_{v\in V} \mathbf{tw}\left(G[B_\ell(v)]\right),
    \end{equation}
    where $\mathbf{tw}(\cdot)$ is the tree-width of a graph.
    \label{def:local-tree-width}
\end{definition}

Ref.~\cite{MarkovShiSimulatingQuantumComputation2008} shows a tensor network simulation algorithm for a circuit on a graph $G$ whose tree-width is growing sublinearly with the number of qubits:

\begin{lemma}[Adapted from Ref.~\cite{MarkovShiSimulatingQuantumComputation2008}]
    Let $P$ be the output distribution of a quantum circuit on a graph $G=(V,E)$ with $T$ gates. Then, there exists a classical algorithm that samples from $P$ in time $T^{O(1)}\exp\left(O(\mathbf{tw}(G))\right)$.
    \label{lem:tree-width-simulation}
\end{lemma}

Using this lemma, we have the generalization of our main result, Theorem~\ref{thm:main_theorem}, for a general graph $G$.

\begin{theorem}
    Let $\mathcal{C}$ be a depth-$d$ noisy random circuit on a graph $G=(V,E)$, and $P$ be the output distribution of $\mathcal{C}(\ketbra{0^n}{0^n})$. Suppose $P$ satisfies the average ${\rm poly}(n)\exp(-\Omega(\ell))$-approximate Markov condition and $\max_{v \in V}|B_\ell(v)|\le f(\ell)$ for all $\ell$, for some function $f$. Then, there exists a classical algorithm that outputs a sample from $P'$ such that $\|P-P'\|_1 \le \varepsilon$ with probability at least $1-\delta$ over the choice of $\mathcal{C}$, in runtime $\exp\left\{O\left[\mathbf{ltw}(G, \log(n/\delta\varepsilon)+O(f(\log(n/\delta\varepsilon))))\right]\right\}$.
    \label{thm:shallow-depth_algorithm-general}
\end{theorem}

\begin{proof}
    By Theorem~\ref{thm:effective_shallow_depth-general}, we can choose $d^* = O(f(\log(n/\delta\varepsilon)))$ such that a random circuit $\mathcal{C}'$ with depth greater than or equal to $d^*$ satisfies $\mathbb{E}_{\mathcal{C}'}\|Q-R\|_1 \le \delta\varepsilon/2$,
    where $Q$ and $R$ are the output distributions of $\mathcal{C}'$ with arbitrary input states, respectively. Then, as in the proof of Theorem~\ref{thm:main_theorem}, we can show that for a distribution $P''$ defined as
    \begin{equation}
        P'' = \prod_{i=1}^{n} P'_{v_i|B'_\ell(v_i)},
    \end{equation}
    we have
    \begin{equation}
        \mathbb{E}_{\mathcal{C}}\|P - P''\|_1 \le \delta\varepsilon,
    \end{equation}
    with $\ell = O(\log(n/\delta\varepsilon))$, where $P'$ is an output distribution of a depth-$d'$ random circuit with $d' = \min\{d, d^*\}$. Further, Markov's inequality gives
    \begin{equation}
        \mathbb{P}\left[\|P-P''\|_1 \ge \varepsilon\right] \le \delta.
    \end{equation}
    In other words, with probability at least $1-\delta$, we have $\|P-P''\|_1 \le \varepsilon$.

    Finally, we give a classical algorithm that samples from $P''$. The algorithm proceeds as follows. Given qubits $v_1, \dots, v_{j-1}$ are sampled from $\prod_{i=1}^{j-1} P'_{v_i|B'_\ell(v_i)}$, let us denote the output in $B'_\ell(v_j)$ as $x\in \{0,1\}^{|B'_\ell(v_j)|}$. Then, the conditional distribution $P'_{v_j|B'_\ell(v_j)}$ can be computed by simulating $\{v_j\}\sqcup B'_\ell(v_j)$ and its lightcone. This corresponds to a depth-$d'$ circuit contained in a subgraph $G[B_{\ell+d'}(v_j)]$. By Lemma~\ref{lem:tree-width-simulation}, we can sample it in time $f(\ell+d')d' \cdot \exp\left(O\left(\mathbf{ltw}(G, \ell+d')\right)\right)$.

    Since $|B'_\ell(v_j)|$ never exceeds $n$, we can always choose $f \le n$. Then, the runtime for sampling $v_j$ is ${\rm poly}(n)\cdot \exp\left(O\left(\mathbf{ltw}(G, \ell+d')\right)\right)$. We repeat this procedure for $j=1, \dots, n$, we can sample from $P''$ in total runtime
    \begin{equation}
        {\rm poly}(n)\cdot \exp\left(O\left(\mathbf{ltw}(G, \ell+d')\right)\right)
        = \exp\left\{O\left[\mathbf{ltw}(G, \log(n/\delta\varepsilon)+O(f(\log(n/\delta\varepsilon))))\right]\right\}.
    \end{equation}
\end{proof}

In particular, if $G$ is a finite dimensional, i.e., $f(\ell) = {\rm poly}(\ell)$, then the runtime of the algorithm is quasipolynomial in $n$, $\varepsilon$, and $\delta$. Specifically, since $\mathbf{ltw}(G, \ell) \le f(\ell)$ for all $\ell$, the runtime of the algorithm is bounded by
\begin{equation}
    \exp\left(O\left(\mathbf{ltw}(G, {\rm polylog}(n/\delta\varepsilon))\right)\right) = \exp\left({\rm polylog}(n/\delta\varepsilon)\right).
\end{equation}

\section{\label{sec:simulation_details}Details of the numerical simulations}

In this section, we present details of the numerical simulations for the 1D Haar random circuits and 2D Clifford circuits. Further simulation results supplementing the main text are also provided.

\begin{figure}
    \centering
    \includegraphics[width=0.7\linewidth]{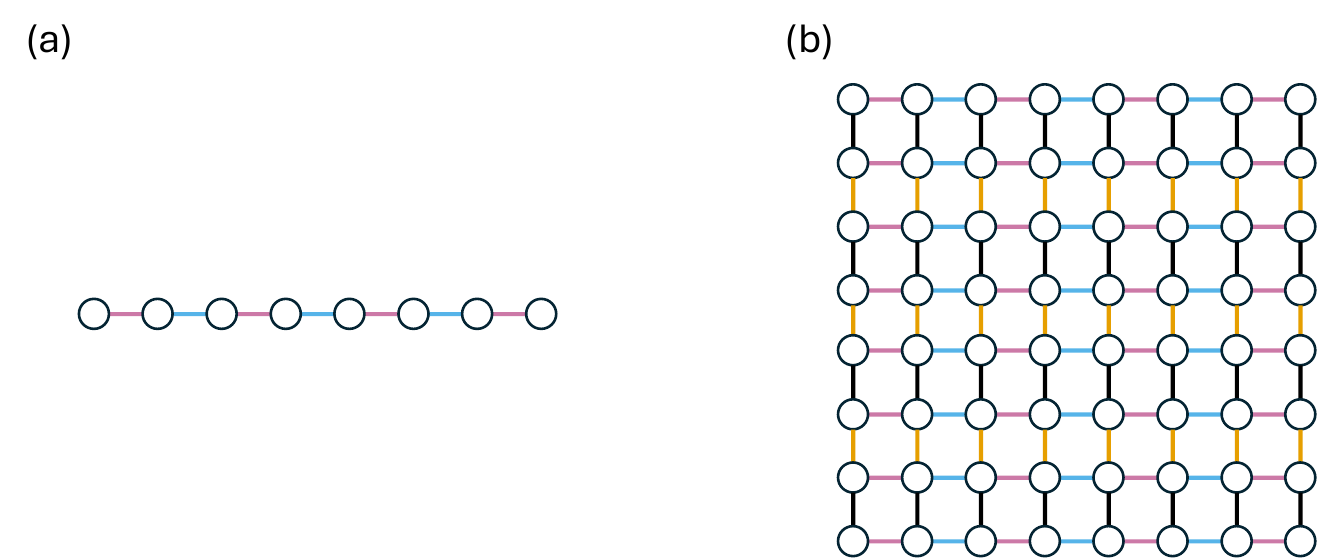}
    \caption{(a) For 1D Haar random circuits, each unitary layer $\mathcal{U}_i$ consists of two-qubit Haar random gates applied on $i$-th and $(i+1)$-th qubits for $i = 1, 3, \dots, n-1$ for even $t$ (qubits connected by purple lines) and $i = 2, 4, \dots, n-2$ for odd $t$ (qubits connected by blue lines). (b) For 2D Clifford random circuits, we apply alternately apply gates on qubits linked with purple lines (when $t = 1 \pmod 4$), orange lines (when $t = 2 \pmod 4$), blue lines (when $t = 3 \pmod 4$), and black lines (when $t = 0 \pmod 4$).}
    \label{fig:gate_patterns}
\end{figure}

\subsection{Methods for 1D Haar random circuits}

We consider 1D Haar random circuits with $n=32$ qubits arranged in a line. Given a depth-$d$ circuit $\mathcal{C} = \mathcal{N}^{\otimes n} \circ \mathcal{U}_d \circ \dots \circ \mathcal{N}^{\otimes n}\circ \mathcal{U}_{1}$, each unitary layer $\mathcal{U}_t$ consists of two-qubit Haar random gates applied on $i$-th and $(i+1)$-th qubits for $i = 1, 3, \dots, n-1$ (resp. $i = 2, 4, \dots, n-2$) when $t$ is even (resp. odd) [Fig.~\ref{fig:gate_patterns}(a)]. We choose the single-qubit noise channel $\mathcal{N}$ to be either an amplitude damping channel,
\begin{equation}
    \mathcal{N}_{\rm amp}(\rho) = K_0 (\rho) K_0^\dagger + K_1 (\rho) K_1^\dagger
\end{equation}
with $K_0 = \begin{pmatrix} 1 & 0 \\ 0 & \sqrt{1-\gamma} \end{pmatrix}$ and $K_1 = \begin{pmatrix} 0 & \sqrt{\gamma} \\ 0 & 0 \end{pmatrix}$, or a depolarizing channel,
\begin{equation}
    \mathcal{N}_{\rm depo}(\rho) = (1-\gamma) (\rho) + \gamma {\rm Tr}[\rho] I/2,
\end{equation}
which are representative noise models for non-unital and unital noise channels, respectively. Here, $0 < \gamma < 1$ denotes the noise rate in both examples.

We simulate the noisy random circuits using the matrix product state (MPS) method~\cite{Vidal2003EfficientClassicalSimulation}. Unlike a typical MPS simulation, which evolves the state vector, we need to keep track of the density matrix throughout the circuit. This is achieved by vectorizing the density matrix and represent it as a matrix product density operator (MPDO)~\cite{VerstraeteMatrixProductDensityOperators2004,Zwolak2004MixedStateDynamics, Noh2020efficientclassical}.

Specifically, we adapted the method in Refs.~\cite{Noh2020efficientclassical,LeeUniversalSpreading2024} with minor improvements. Consider a generic $n$-qubit density matrix $\rho$ written as
\begin{equation}
    \rho = \sum_{i_1,j_1, \dots, i_n,j_n} \rho_{i_1 j_1, \dots, i_n j_n} \ketbra{i_1, \dots, i_n}{j_1, \dots, j_n},
\end{equation}
where the indices $i_k$ and $j_k$ run over the computational basis $\{0,1\}$ for each qubit $k=1,\dots,n$. We can vectorize this density matrix by mapping each basis $\ketbra{i_k}{j_k}$ to a vector $\Ket{I_k}$ in a 4-dimensional space where $I_k = 2i_k + j_k$. Then, we can write the vectorized density matrix $\Ket{\rho}$ as
\begin{equation}
    \Ket{\rho} = \sum_{I_1, \dots, I_n} \rho_{I_1, \dots, I_n} \Ket{I_1, \dots, I_n}.
\end{equation}
Just like a state vector, MPDO represents a large-dimensional tensor $\rho_{I_1, \dots, I_n}$ as a product of small matrices $A^{[k]I_k}$ for each qubit $k=1,\dots,n$ and each index $I_k$:
\begin{equation}\label{eq:mpdo}
    \rho_{I_1, \dots, I_n} = \sum_{\alpha_1,\dots,\alpha_{n-1}} A^{[1]I_1}_{\alpha_1} A^{[2]I_2}_{\alpha_1\alpha_2} \dots A^{[n-1]I_{n-1}}_{\alpha_{n}\alpha_{n-1}} A^{[n], I_n}_{\alpha_{n-1}},
\end{equation}
where $\alpha_k$ are the bond indices connecting the different sites in the MPDO representation. For example, the initial state $\rho = \ketbra{0^n}{0^n}$ can be trivially represented with $1\times 1$ matrices
\begin{equation}
    A^{[k]I_k} = \begin{cases}
        \begin{pmatrix} 1 \end{pmatrix} & \text{if } I_k = 0,\\
        \begin{pmatrix} 0 \end{pmatrix} & \text{otherwise.}
    \end{cases}
\end{equation}
for $k=1,\dots,n$.

\subsubsection{MPDO update rules}

Beginning with $\ketbra{0^n}{0^n}$ state, the simulation proceeds by applying the unitary gates and the noise channels in alternating order. For simplicity, we denote a unitary channel $\mathcal{U}$ acting on $k$-th and $(k+1)$-th qubits combined with the noise channel as $\mathcal{M}$:
\begin{align}
    \mathcal{M} &= \mathcal{N}^{\otimes 2} \circ \mathcal{U}\\
    &= \sum_{I_k,I_{k+1},J_k,J_{k+1}} M_{I_k,I_{k+1},J_k,J_{k+1}} \Ket{I_k, I_{k+1}}\Bra{J_k, J_{k+1}}.
\end{align}
Here, denoting the vectorized indices as
\begin{align}
    I_k &= 2i_k + j_k,\\
    J_k &= 2i'_k + j'_k,\\
    I_{k+1} &= 2i_{k+1} + j_{k+1},\\
    J_{k+1} &= 2i'_{k+1} + j'_{k+1},
\end{align}
the matrix element $M_{I_k,I_{k+1},J_k,J_{k+1}}$ is given by
\begin{equation}
    M_{I_k,I_{k+1},J_k,J_{k+1}} = \bra{i_k, i_{k+1}}\mathcal{M}\left(\ketbra{i'_k i'_{k+1}}{j'_k j'_{k+1}}\right)\ket{j_k, j_{k+1}}.
\end{equation}

Before applying the channel $\mathcal{M}$, we convert MPDO in Eq.~\eqref{eq:mpdo} into the following canonical form for optimizing the numerical cost:
\begin{equation}\label{eq:canonical_form}
    \rho_{I_1, \dots, I_n} = \sum_{\alpha_1,\dots,\alpha_{n-1}} L^{[1]I_1}_{\alpha_1} L^{[2]I_2}_{\alpha_1\alpha_2} \dots L^{[k]I_k}_{\alpha_{k-1}, \alpha_{k}} \lambda^{[k]}_{\alpha_k} R^{[k+1]I_{k+1}}_{\alpha_{k}, \alpha_{k+1}} \dots R^{[n-1]I_{n-1}}_{\alpha_{n}\alpha_{n-1}} R^{[n], I_n}_{\alpha_{n-1}},
\end{equation}
where the matrices $L^{[1]I_1}, \dots, L^{[k]I_k}$ and $R^{[k+1]I_{k+1}}, \dots, R^{[n]I_n}$ satisfy
\begin{align}
    \sum_{I_1} L^{[1]I_1}_{\alpha_1} \left(L^{[1]I_1}_{\alpha'_1}\right)^* &= \delta_{\alpha_1, \alpha'_1},\\
    \sum_{\alpha_{l-1}, I_l} L^{[l]I_l}_{\alpha_{l-1}, \alpha_l}\left(L^{[l]I_l}_{\alpha_{l-1}, \alpha'_l}\right)^* &= \delta_{\alpha_l, \alpha'_l} \quad \text{for } l=2,\dots,k,\\
    \sum_{\alpha_{l-1}, I_l} R^{[l]I_l}_{\alpha_{l-1}, \alpha_l}\left(R^{[l]I_l}_{\alpha_{l-1}, \alpha'_l}\right)^* &= \delta_{\alpha_{l-1}, \alpha'_{l-1}} \quad \text{for } l=k+1,\dots,n-1,\\
    \sum_{I_1} R^{[n]I_n}_{\alpha_{n-1}} \left(R^{[n]I_n}_{\alpha'_{n-1}}\right)^* &= \delta_{\alpha_{n-1}, \alpha'_{n-1}},
\end{align}
and the elements of the vector $\lambda^{[k]}$ are non-negative real numbers. The conversion to this canonical form can be done by successive applications of singular value decompositions (SVDs). See, e.g., Sec.~4 in Ref.~\cite{SCHOLLWOCK201196} for the details of this canonical form.

After converting the MPDO to the canonical form, we apply the channel $\mathcal{M}$ to $\Ket{\rho}$. Denoting $\Ket{\rho'} = \sum_{I_1, \dots, I_n} \rho'_{I_1, \dots, I_n}\Ket{I_1\dots I_n}$ as the density matrix after applying $\mathcal{M}$, we have
\begin{equation}
    \rho'_{I_1, \dots, I_n} = \sum_{\alpha_1,\dots,\alpha_{k-1}, \alpha_{k+1},\dots,\alpha_{n-1}} L^{[1]I_1}_{\alpha_1} L^{[2]I_2}_{\alpha_1\alpha_2} \dots L^{[k-1]I_{k-1}}_{\alpha_{k-2}\alpha_{k-1}} A^{I_k, I_{k+1}}_{\alpha_{k-1}, \alpha_{k+1}} R^{[k+2]I_{k+2}}_{\alpha_{k+1}\alpha_{k+2}} \dots R^{[n-1]I_{n-1}}_{\alpha_{n}\alpha_{n-1}} R^{[n], I_n}_{\alpha_{n-1}},
\end{equation}
where the matrix $A^{I_k, I_{k+1}}_{\alpha_{k-1}, \alpha_{k+1}}$ is given by
\begin{equation}
    A^{I_k, I_{k+1}}_{\alpha_{k-1}, \alpha_{k+1}} = \sum_{J_k, \alpha_k ,J_{k+1}} M_{I_k,I_{k+1},J_k,J_{k+1}} L^{[k]J_{k}}_{\alpha_{k-1}, \alpha_{k}} \lambda^{[k]}_{\alpha_{k}} R^{[k+1]J_{k+1}}_{\alpha_{k}, \alpha_{k+1}}.
\end{equation}
We further perform SVDs on the matrices $A^{I_k, I_{k+1}}_{\alpha_{k-1}, \alpha_{k+1}}$, i.e.,
\begin{equation}\label{eq:svd_update}
    A^{I_k, I_{k+1}}_{\alpha_{k-1}, \alpha_{k+1}} = \sum_{\alpha_k} L'^{[k]I_k}_{\alpha_{k-1}, \alpha_k} \lambda'^{[k]}_{\alpha_k} R'^{[k+1]I_{k+1}}_{\alpha_{k}, \alpha_{k+1}}.
\end{equation}
to convert the MPDO back to the canonical form in Eq.~\eqref{eq:canonical_form}.

If the bond index $\alpha_k$ runs over $0, \dots, d_k-1$, the updated bond index $\alpha'_k$ can run over $0, \dots, 4d_k-1$. Therefore, for an exact simulation, we need to keep track of matrices with exponentially growing number of elements. To avoid this, we truncate the bond indices to a fixed size $\chi$ for $\alpha_1,\dots,\alpha_{n-1}$. This is done by keeping only the largest $\chi$ singular values $\lambda'^{[k]}_{\alpha_k}$ in Eq.~\eqref{eq:svd_update} and discarding the rest. The truncation is performed after every application of the channel $\mathcal{M}$, and the bond indices run only over $0, \dots, \chi -1$. Here, $\chi$ is called the bond dimension. With this truncation scheme, the computational cost of the simulation is $O(nd\chi^3)$.

For implementation, we used a Python tensor network package \texttt{quimb}~\cite{quimb2018}, and the source code for our simulation is available upon request.

\subsubsection{Benchmarking MPDO simulation}

\begin{figure}
    \centering
    \includegraphics[width=\linewidth]{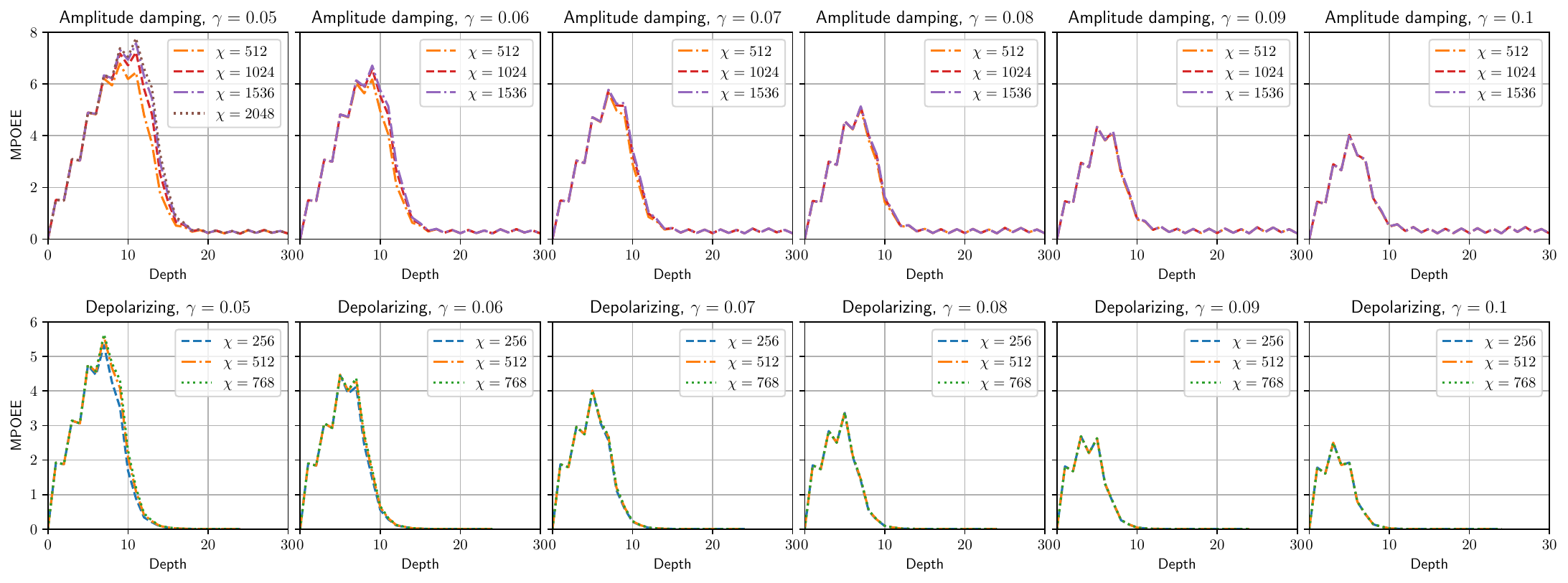}
    \caption{Benchmarking the performance of the classical simulation methods for 1D Haar random circuits. Plots in the upper panel present the MPOEE with various choices of the bond dimension $\chi$ for different noise rates $\gamma$, as a function of the circuit depth. The lower panel presents the same quantities but for the circuits with depolarizing channel. The MPOEE is calculated across the cut between the $16$-th and $17$-th qubits, and averaged over $4$ different circuit realizations.}
    \label{fig:numerical_benchmark}
\end{figure}

For MPS simulations for a state vector, the above truncation scheme is optimal in the sense that it minimizes the error in $\ell^2$ distance~\cite{SCHOLLWOCK201196}. However, much less is known for the MPDO simulation for evolving density matrices---while one needs to bound the error in trace distance, only a few measures of the error are known (e.g., \cite{guth2020efficientdescriptionMPDO}) that are not directly computable in practice.

To benchmark the performance of the MPDO simulation, therefore, we rely on the following heuristic argument. In a nutshell, the bond dimension $\chi$ controls how much correlation we can capture in the MPDO representation, and one need to choose exponentially large $\chi$ to capture the correlations to simulate an arbitrary state. To quantify the amount of correlation captured in the MPDO, we consider matrix product operator entanglement entropy (MPOEE)~\cite{Zanardi2001EntanglementQuantumEvolutions,Noh2020efficientclassical}. MPOEE across the cut between the $k$-th and $(k+1)$-th qubits is defined as
\begin{equation}
    S_k = -\sum_{\alpha_k=0}^{\chi-1} \frac{(\lambda^{[k]}_{\alpha_k})^2}{\sum_{\alpha'_k}(\lambda^{[k]}_{\alpha'_k})^2} \log_2 \left(\frac{(\lambda^{[k]}_{\alpha_k})^2}{\sum_{\alpha'_k}(\lambda^{[k]}_{\alpha'_k})^2}\right),
\end{equation}
where $\lambda^{[k]}_{\alpha_k}$ are the singular values obtained from the mixed canonical form in Eq.~\eqref{eq:canonical_form}. The MPOEE $S_k$ quantifies the amount of correlation between the qubits in $\{1,\dots, k\}$ and those in $\{k+1,\dots, n\}$. Therefore, if the bond dimension $\chi$ is too small, the MPDO representation cannot capture the correlation and thus increasing the bond dimension would result in a larger MPOEE. As we keep increasing $\chi$, the MPOEE will eventually saturate at a certain value, which is determined by the amount of correlation in the original state. In other words, if we choose the MPOEE to be large enough to represent the original state, further increasing $\chi$ would not change the MPOEE. Using this insight, we use the convergence of MPOEE as a proxy for the performance of the MPDO simulation.

For both amplitude damping and depolarizing noise channels, we benchmark the performance of the MPDO simulation by computing the MPOEE for various bond dimensions $\chi$ and noise rates $\gamma$. While MPOEE can be calculated across different cuts in the system, we focus on the representative cut between the equal bipartitions of the qubits, i.e., between the $16$-th and $17$-th qubits. We repeat the simulation for $4$ different circuit realizations and take the average the MPOEE over them. The results are shown in Fig.~\ref{fig:numerical_benchmark}. For the amplitude damping channel, we observe that the the bond dimension of $\chi = 1536$ is sufficient for $\gamma = 0.05$ and $\chi = 1024$ is sufficient for $\gamma = 0.06\text{--}0.1$, as further increasing $\chi$ does not change the MPOEE. For the depolarizing channel, we observe that the bond dimension of $\chi = 512$ is sufficient for $\gamma = 0.05\text{--}0.1$. Based on these results, we choose $\chi = 1536$ (resp. $\chi = 1024$) for the amplitude damping channel with $\gamma = 0.05$ (resp. $\gamma = 0.06\text{--}0.1$), and $\chi = 512$ for the depolarizing channel with $\gamma = 0.05\text{--}0.1$ in the numerical results presented in the main text.

\subsubsection{Estimating conditional mutual information}

Having obtained MPDO of Eq.~\eqref{eq:mpdo} that well-approximates the output density matrix $\rho$, we estimate CMI $I(X:Z|Y)$ of the corresponding output distribution $P$. To this end, we first obtain the output distribution $P$ by getting rid of the off-diagonal elements of the density matrix, i.e., plugging $I_k = 2i_k + i_k$ into Eq.~\eqref{eq:mpdo}:
\begin{align}
    P_{i_1,\dots,i_n} &= \rho_{3i_1,\dots,3i_n}\\
    &= \sum_{\alpha_1,\dots,\alpha_{n-1}} A^{[1]3i_1}_{\alpha_1} A^{[2]3i_2}_{\alpha_1\alpha_2} \dots A^{[n-1]3i_{n-1}}_{\alpha_{n}\alpha_{n-1}} A^{[n], 3i_n}_{\alpha_{n-1}},\label{eq:TT}
\end{align}
where $i_k \in \{0,1\}$ for $k=1,\dots,n$. We can compute CMI,
\begin{equation}
    I(X:Z|Y) = H(XY) + H(YZ) - H(XYZ) - H(Y),
\end{equation}
by estimating the Shannon entropies $H(XY)$, $H(YZ)$, $H(XYZ)$, and $H(Y)$ from the output distribution $P$. However, directly computing Shannon entropies from a multi-dimensional distribution can be challenging, even when the distribution is represented as a matrix product form of Eq.~\eqref{eq:TT}.

To address this issue, we develop a Monte Carlo (MC) method that leverages the matrix product representation of the output distribution. The key idea is that while the full distribution $P$ is high-dimensional, we can decompose the entropies into sums of lower-dimensional marginals. Specifically, if $X$ contains $k$ qubits with $X_1,\dots, X_k$ denoting their measurement outcomes, we can express the Shannon entropy $H(X)$ as
\begin{align}
    H(X) &= H(X_1,\dots,X_k)\\
    &= H(X_1) + H(X_2|X_1) + H(X_3|X_{<3})\dots + H(X_k|X_{<k})
\end{align}
where $X_{<i} = X_1 \dots X_{i-1}$. Since each conditional entropy is written as $H(X_i|X_{<i}) = \underset{x_{<i} \sim P_{X_{<i}}}{\mathbb{E}}\left[H(X_i|X_{<i}=x_{<i})\right]$, we have
\begin{equation}
    H(X) = \underset{(x_1,\dots,x_k) \sim P_X}{\mathbb{E}}\left[\sum_{i=1}^{k}H(X_i|X_{<i}=x_{<i})\right].
\end{equation}
Thus, we can estimate $H(X)$ with $\hat{H}(X)$ with $N$ MC samples $x^{(1)}, \dots, x^{(N)} \in \{0,1\}^k$ from the output distribution $P$:
\begin{equation}
    \hat{H}(X) = \frac{1}{N}\sum_{(x_1,\dots,x_k) \sim P_X} \sum_{i=1}^{k} H(X_i|X_{<i}=x_{<i}),
\end{equation}
where $H(X_i|X_{<i}=x_{<i})$ is computed from the conditional distribution $P_{X_i|X_{<i}=x_{<i}}$. Note that calculating $H(X_i|X_{<i}=x_{<i})$ only contains two probability masses,
\begin{equation}
    H(X_i|X_{<i}=x_{<i}) = -\sum_{x_i \in \{0,1\}} P_{X}(X_i = x_i|X_{<i}=x_{<i}) \log_2 P_{X}(X_i = x_i|X_{<i}=x_{<i}).
\end{equation}
Since the matrix product form of the output distribution $P$ allows us to efficiently compute the marginal and conditional probabilities, we can obtain $\hat{H}(X)$ efficiently.

With this MC method, we estimated the CMI $I(X:Z|Y)$ for the 1D Haar random circuits in the main text. For each circuit realization and depth, we used $N=1{,}000$ MC samples from the output distribution $P$. The final CMI is obtained by averaging over $64$ different circuit realizations.

\subsection{Methods for 2D Clifford circuits}

We consider 2D Clifford circuits with $n=32 \times 32$ qubits arranged in a square lattice. Labeling the qubits as $(i,j)$ for $i,j = 1, \dots, 32$, we consider a depth-$d$ circuit $\mathcal{C} = \mathcal{N}^{\otimes n} \circ \mathcal{U}_d \circ \dots \circ \mathcal{N}^{\otimes n}\circ \mathcal{U}_{1}$, where each unitary layer $\mathcal{U}_t$ consists of two-qubit random Clifford gates applied on the nearest-neighbor pairs of qubits. Specifically, for $t$-th layer, we apply the two-qubit Clifford gates according to the following patterns [Fig.~\ref{fig:gate_patterns}(b)]:
\begin{enumerate}
    \item[(i)] If $t = 1 \pmod 4$, apply gates on horizontal pairs $(i,j)$ and $(i+1,j)$ for $i = 1, 3, \dots, 31$ and $j=1, 2, \dots, 32$.
    \item[(ii)] If $t = 2 \pmod 4$, apply gates on vertical pairs $(i,j)$ and $(i,j+1)$ for $i = 1, 2, \dots, 32$ and $j=2, 4, \dots, 30$.
    \item[(iii)] If $t = 3 \pmod 4$, apply gates on horizontal pairs $(i,j)$ and $(i+1,j)$ for $i = 2, 4, \dots, 30$ and $j=1, 2, \dots, 32$.
    \item[(iv)] If $t = 0 \pmod 4$, apply gates on vertical pairs $(i,j)$ and $(i,j+1)$ for $i = 1, 2, \dots, 32$ and $j=1, 3, \dots, 31$.
\end{enumerate}
The single-qubit noise channel $\mathcal{N}$ is chosen to be a heralded reset channel or heralded depolarizing channel with noise rate $\gamma$, which are defined as follows:
\begin{align}
    \mathcal{N}_{\rm hreset}(\rho) &= \begin{cases}
        \ketbra{0}{0} & \text{with probability } \gamma,\\
        \rho & \text{with probability } 1-\gamma,
    \end{cases}\\
    \mathcal{N}_{\rm hdepo}(\rho) &= \begin{cases}
        I/2 & \text{with probability } \gamma,\\
        \rho & \text{with probability } 1-\gamma.
    \end{cases}
\end{align}

We simulate the noisy random circuits using the stabilizer formalism~\cite{gottesman1998heisenbergrepresentationquantumcomputers,AaronsonGottesmanImprovedSimulationStabilizer2004}, which is efficient for Clifford circuits with the noise channels described above. The stabilizer formalism allows us to succinctly represent the state of the circuit with an associated stabilizer group $\mathcal{S}$, which is an abelian subgroup of the $n$-qubit Pauli group with the constraint of $-I \notin \mathcal{S}$. Then, the corresponding state $\rho$ is given by
\begin{equation}
    \rho = \frac{1}{2^n} \sum_{s \in \mathcal{S}} s.
\end{equation}
In practice, we only keep track of the generating set $\mathcal{G}$ of the stabilizer group $\mathcal{S}$. Since $\mathcal{S}$ is an abelian subgroup of $n$-qubit Pauli group, the generating set contains at most $n$ elements, which is efficient to store and manipulate. With these generators, we $\rho$ can be alternately written as
\begin{equation}
    \rho = \prod_{g \in \mathcal{G}} \left(\frac{1+g}{2}\right).
\end{equation}

For implementation, we used a Julia package \texttt{QuantumClifford.jl}~\cite{QuantumCliffordjl}, and the source code for our simulation is available upon request.

\subsubsection{Stabilizer update rules}

Starting with the initial state $\rho = \ketbra{0^n}{0^n}$ that corresponds to the stabilizer group generated by $\mathcal{G} = \{Z_1, Z_2, \dots, Z_n\}$, we simulate the noisy Clifford random circuits by updating the stabilizer group accordingly. Given a state $\rho$ with the stabilizer group $\mathcal{S}$, applying a Clifford gate $U$ results in
\begin{equation}
    \rho' = U \rho U^\dagger = \frac{1}{2^n} \sum_{s \in \mathcal{S}} U s U^\dagger,
\end{equation}
which corresponds to the new stabilizer group $\mathcal{S}' = \{UsU^\dagger : s \in \mathcal{S}\}$. Since a Clifford gate maps Pauli operators to Pauli operators, $\mathcal{S}'$ also forms a stabilizer group.

The update rule for the heralded depolarizing channel is also efficiently described in stabilizer formalism. Suppose that the $i$-th qubit is heralded to be depolarized, i.e., $\rho' = {\rm Tr}_i[\rho] \otimes I_i/2$, where $\rho_i$ is the reduced density matrix of the $i$-th qubit. Then, the resulting $\rho'$ is given by
\begin{align}
    \rho' &= \frac{1}{2^n} \sum_{s \in \mathcal{S}} {\rm Tr}_i [s] \otimes I_i/2\\
    &= \frac{1}{2^n} \left(\sum_{s \in \mathcal{S}: s_i = I} {\rm Tr}_i [s] + \sum_{s \in \mathcal{S}: s_i \ne I} {\rm Tr}_i [s] \right)\otimes I_i/2 \\
    &= \frac{1}{2^n} \sum_{s \in \mathcal{S}: s_i =I} s,
\end{align}
where $s_i$ is the Pauli operator that $s$ acts on the $i$-th qubit. Therefore, the resulting state $\rho'$ corresponds to
\begin{equation}
    \mathcal{S}' = \{s \in S: s_i = I\},
\end{equation}
which also forms a stabilizer group. The effect of the heralded reset channel can be described similarly. Suppose that the $i$-th qubit is heralded to be reset, i.e., $\rho' = {\rm Tr}_i[\rho] \otimes \ketbra{0}{0}_i/2$, where $\rho_i$ is the reduced density matrix of the $i$-th qubit. Then, $\rho'$ is given by
\begin{align}
    \rho' &= \frac{1}{2^n} \sum_{s \in \mathcal{S}} {\rm Tr}_i [s] \otimes \ketbra{0}{0}_i\\
    &= \frac{1}{2^n} \left(\sum_{s \in \mathcal{S}: s_i = I} {\rm Tr}_i [s] + \sum_{s \in \mathcal{S}: s_i \ne I} {\rm Tr}_i [s] \right)\otimes \frac{I_i + Z_i}{2} \\
    &= \frac{1}{2^n} \sum_{s \in \mathcal{S}: s_i =I} (s + s Z_i),
\end{align}
which corresponds to the stabilizer group
\begin{equation}
    \mathcal{S}' = \{s \in S: s_i = I\} \cup \{s Z_i: s \in S, s_i = I \}.
\end{equation}
In this way, we obtain the the stabilizer group $\mathcal{S}$ corresponding to the output state of the noisy Clifford circuits.

\subsubsection{Entropy of the output distribution}

Once we obtain the stabilizer group $\mathcal{S}$ corresponding to the output state $\rho$, we calculate the conditional mutual information $I(X:Z|Y)$ of the output distribution $P$. First, the output distribution $P$ is expressed in terms of the density matrix $\rho'$ which is obtained by getting rid of the off-diagonal elements of $\rho$. Specifically, given that $\rho$ corresponds to the stabilizer group $\mathcal{S}$, we can write $\rho'$ as $\rho' = \frac{1}{2^n} \sum_{s \in \mathcal{S}'} s$ where
\begin{equation}
    \mathcal{S}' = \{s \in \mathcal{S}: s_i = I \text{ or } s_i = Z \text{ for } i=1,\dots,n\},
\end{equation}
i.e., the stabilizer group $\mathcal{S}'$ contains only the Pauli operators that are diagonal in the computational basis.

Since $\rho'$ is a diagonal matrix and each diagonal element corresponds to a probability $P$, Shannon entropy of $P$ is equal to the von Neumann entropy of $\rho'$. Furthermore, the von Neumann entropy $S(\rho')$ of $\rho'$ can be computed from the number of elements in the stabilizer group $\mathcal{S}'$~\cite{Audenaert_2005mixedstab}, given by
\begin{equation}\label{eq:vn_entropy}
    S(\rho') = n - \log_2 |\mathcal{S}'|.
\end{equation}
Incorporating the above results, we can compute the Shannon entropy $H(X)$ of the output distribution $P$. First, we obtain the reduced density matrix $\rho'_{X}$ by
\begin{align}
    \rho'_{X} &= \frac{1}{2^n} \sum_{s \in \mathcal{S}'} {\rm Tr}_{[n]\setminus X}[s]\\
    &= \frac{1}{2^n} \sum_{s \in \mathcal{S}': s_i = I \,\forall i \in [n] \setminus X} {\rm Tr}_{[n]\setminus X}[s]\\
    &= \frac{1}{2^{n-|X|}} \sum_{s \in \mathcal{S}': s_i = I \,\forall i \in [n] \setminus X} s_{X},
\end{align}
where $s_{X}$ is the restriction of $s$ to the qubits in $X$. Then, $\rho'_X$ corresponds to the stabilizer group
\begin{equation}
    \mathcal{S}'_X = \{s_X: s \in \mathcal{S}',\ s_i = I\ \forall i \in [n] \setminus X\}.
\end{equation}
By using Eq.~\eqref{eq:vn_entropy} and the relation $H(X) = S(\rho'_X)$, we can compute the CMI as
\begin{align}
    I(X:Z|Y) &= H(XY) + H(YZ) - H(XYZ) - H(Y)\\
    &= -\log_2 |\mathcal{S}'_{XY}| - \log_2 |\mathcal{S}'_{YZ}| + \log_2 |\mathcal{S}'_{XYZ}| + \log_2 |\mathcal{S}'_{Y}|.
\end{align}

\subsection{\label{sec:additional_numerics}Additional numerical results}

In the main text, we presented the numerical results of decaying CMI $I(X:Z|Y)$ for particular choices of $X$, $Y$, and $Z$ in both 1D Haar random circuits and 2D Clifford circuits. Here, we present additional numerical results for various sizes of $X$, $Y$, and $Z$ in both cases. In the main text, we presented the numerical results averaged over two (resp. four) consecutive depths for 1D Haar random circuits (resp. 2D Clifford circuits) thereby smoothing out the fluctuations in the CMI. Here, we present the numerical results for individual depths for both cases.

\begin{figure}
    \centering
    \includegraphics[width=\textwidth]{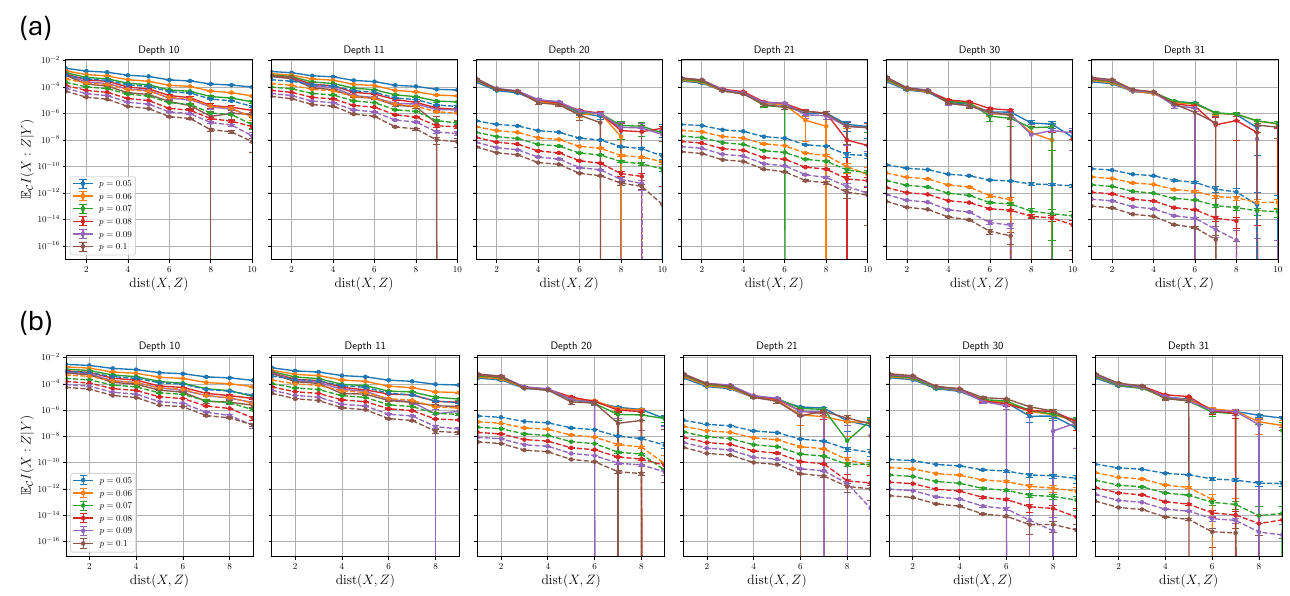}
    \caption{\label{fig:additional_numerics_1d}Additional numerical results for decaying CMI in 1D Haar random circuits where (a) $X$ is two qubits in the middle, and (b) $X$ is four qubits in the middle. For both plots, CMI is estimated with $1{,}000$ MC samples, and further averaged over $64$ circuit realizations.}
\end{figure}

\begin{figure}
    \centering
    \includegraphics[width=\textwidth]{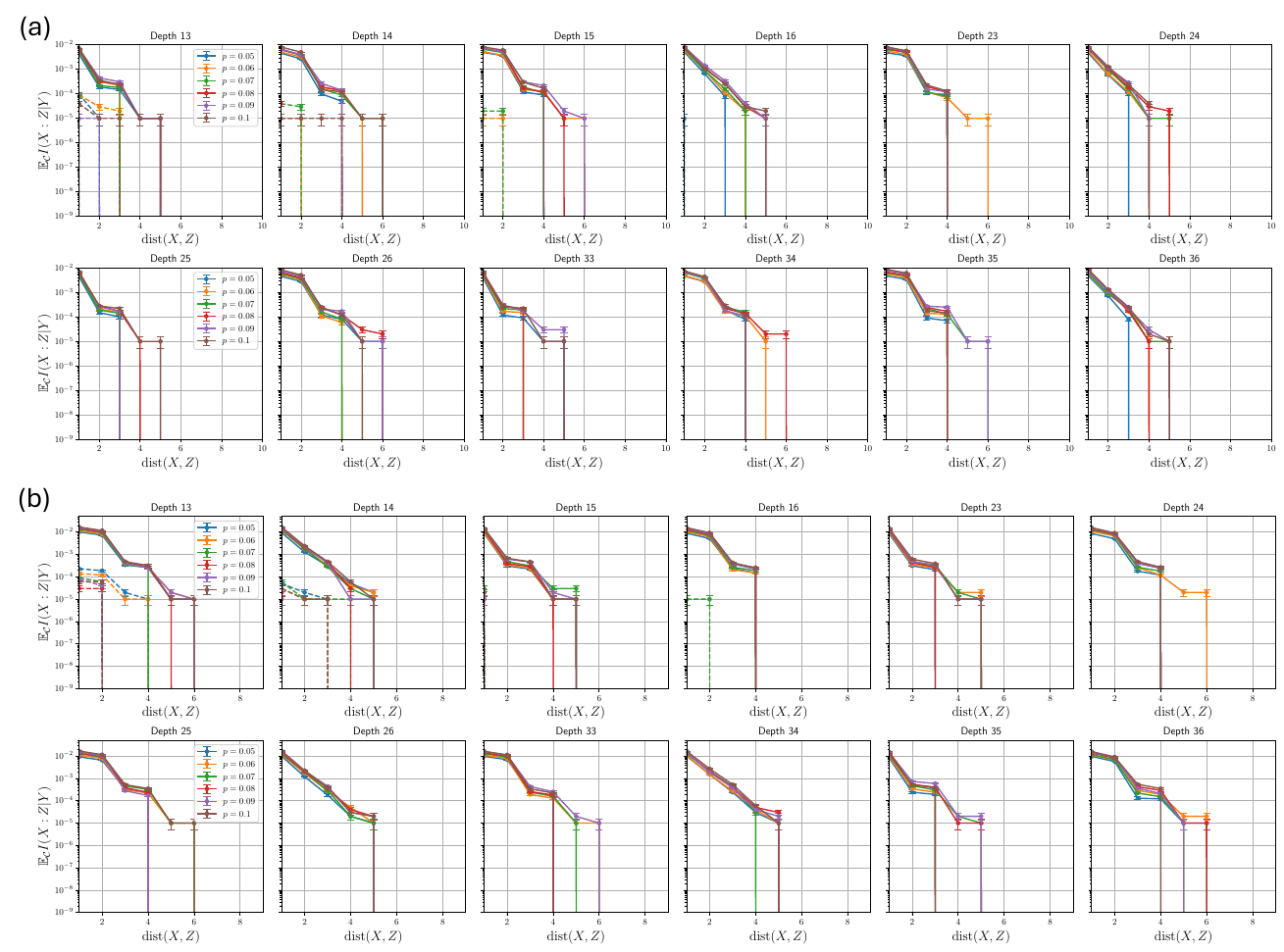}
    \caption{\label{fig:additional_numerics_2d}Additional numerical results for decaying CMI in 2D Clifford circuits where (a) $X$ is $2\times 2$ qubits in the center, and (b) $X$ is $4 \times 4$ qubits in the middle. For both plots, CMI is averaged over $100{,}000$ circuit realizations.}
\end{figure}

Fig.~\ref{fig:additional_numerics_1d} presents the additional numerical results for decaying CMI in 1D Haar random circuits, choosing $X$ as the two qubits in the middle [Fig.~\ref{fig:additional_numerics_1d}(a)] or four qubits in the middle [Fig.~\ref{fig:additional_numerics_1d}(b)]. All the results show a clear decay of CMI with the distance between $X$ and $Z$.

Fig.~\ref{fig:additional_numerics_2d} presents the results for decaying CMI in 2D Clifford circuits, choosing $X$ as the $2\times 2$ qubits in the center [Fig.~\ref{fig:additional_numerics_2d}(a)] or $4 \times 4$ qubits in the middle [Fig.~\ref{fig:additional_numerics_2d}(b)]. While the results for the heralded depolarizing channel is too small to be reliably estimated, the results for the heralded reset channel show a clear decay of CMI with the distance between $X$ and $Z$.

\section{\label{sec:effective_depth}Further discussion on the effective depth of noisy random circuits}

\begin{figure}
    \centering
    \includegraphics[width=\textwidth]{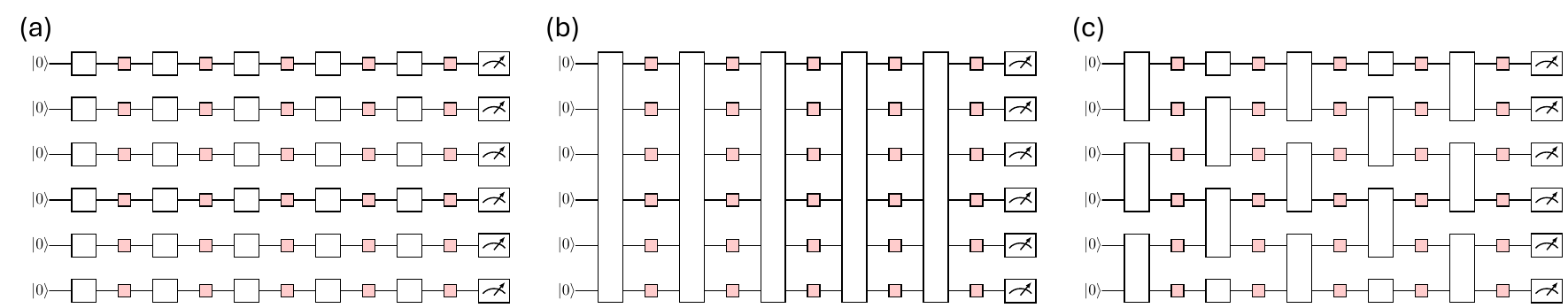}
    \caption{\label{fig:apdx-effective-depth}Illustration of different noisy random circuit models with alternating layers of unitary gates (white boxes) and noise channels (red boxes). (a) Random circuit with single-qubit gates. (b) Random circuit with global gates. (c) Random circuit with two-qubit gates applied on nearest-neighbor pairs of qubits.}
\end{figure}

In this section, we remark that while we establish that noisy random circuits become effectively shallow with $d^* = \mathrm{polylog}(n)$, it is expected that this can be improved to $d^* = O(\log n)$. To support this conjecture, we analyze two additional models of random circuits---single-qubit gate circuits and global gate circuits---in addition to our two-qubit gate model.

Specifically, consider a $n$-qubit depth-$d$ random circuit $\mathcal{C}$ consisting of alternating layers of unitary gates and noise channels,
\begin{equation}
    \mathcal{C} = \mathcal{N}^{\otimes n} \circ \mathcal{U}_d \circ \cdots \circ \mathcal{N}^{\otimes n} \circ \mathcal{U}_1,
\end{equation}
where in the single-qubit gate model, each $\mathcal{U}_t$ consists of independent Haar-random single-qubit unitaries [Fig.~\ref{fig:apdx-effective-depth}(a)], and in the global gate model, each $\mathcal{U}_t$ is an $n$-qubit Haar-random unitary acting on all qubits [Fig.~\ref{fig:apdx-effective-depth}(b)]. Our two-qubit gate model, in which each $\mathcal{U}_t$ consists of Haar-random two-qubit gates applied to nearest-neighbor pairs, interpolates between these two extremes [Fig.~\ref{fig:apdx-effective-depth}(c)].

First, we note that the effective depth of the random circuit with single-qubit gates is $d^* = O(\log n)$. To see this, consider a depth-$d$ circuit $\mathcal{C}$ and a truncated depth-$d^*$ circuit $\mathcal{C}'$ consisting of the last $d^*$ layers of $\mathcal{C}$. Let $P$ and $Q$ denote the output distributions of $\mathcal{C}(\ketbra{0^n}{0^n})$ and $\mathcal{C}'(\ketbra{0^n}{0^n})$, respectively. Since all gates are single-qubit ones, we have
\begin{equation}
    P = \prod_{j=1}^n P_j, \qquad Q = \prod_{j=1}^n Q_j,
\end{equation}
where $P_j$ and $Q_j$ are the output distributions for qubit $j$ in $\mathcal{C}(\ketbra{0^n}{0^n})$ and $\mathcal{C}'(\ketbra{0^n}{0^n})$, respectively.

We now remark that Proposition~\ref{prop:mele} that is from Ref.~\cite{mele2024noiseinducedshallowcircuitsabsence} is applicable to this single-qubit random circuit model, and thus Lemma~\ref{lem:indistinguishable_marginals} holds. Therefore, for each $j=1,\ldots,n$, we have
\begin{equation}
    \mathbb{E}_{\mathcal{C}}\|P_j - Q_j\|_{1} \le 2\exp(-\Omega(d^*)).
\end{equation}

Moreover,
\begin{align}
    \|P - Q\|_1 &= \left\|\prod_{j=1}^n P_j - \prod_{j=1}^n Q_j\right\|_1\\
    &\le \left\|P_1\prod_{j=2}^{n} P_j - Q_1\prod_{j=2}^n P_j\right\|_1 + \left\|Q_1\prod_{j=2}^{n} P_j - Q_1\prod_{j=2}^n Q_j\right\|_1\\
    &= \left\|P_1 - Q_1\right\|_1 + \left\|\prod_{j=2}^{n} P_j - \prod_{j=2}^n Q_j\right\|_1,
\end{align}
and by repeating this procedure for the remaining qubits ($j=2,\ldots,n$), we have $\|P - Q\|_1 \le \sum_{j=1}^n \|P_j - Q_j\|_1$. This leads to
\begin{equation}
    \mathbb{E}_{\mathcal{C}}\|P - Q\|_1 \le 2n\exp(-\Omega(d^*)),
\end{equation}
and therefore we can make $\|P - Q\|_1$ arbitrarily small by taking $d^*=O(\log n)$.

Second, for the random circuit with global gates, Ref.~\cite{shtanko2024complexitylocalquantumcircuits} shows that the effective depth is $d^* = O(1)$ at least when the noise channel is a generalized replacement channel,
\begin{equation}
    \mathcal{N}_{\mathrm{rep}; \sigma}(\rho) = (1-\gamma)\rho + \gamma\sigma,
\end{equation}
where $\sigma$ is an arbitrary single-qubit state and $0<\gamma\le 1$ is the noise rate. This channel may be unital ($\sigma=I/2$) or non-unital ($\sigma\neq I/2$). In particular:

\begin{theorem}[Adapted from Ref.~\cite{shtanko2024complexitylocalquantumcircuits}]
    Let $\mathcal{C}$ be a depth-$d$ random circuit with global Haar-random gates and generalized replacement noise with the noise rate $\gamma$. Then for arbitrary $n$-qubit states $\rho$ and $\sigma$,
    \begin{equation}
        \mathbb{E}_{\mathcal{C}}\|\mathcal{C}(\rho) - \mathcal{C}(\sigma)\|_1 =
        O\left(2^{n/2}\left(1-\tfrac{\gamma}{2}\right)^{n(d-1)/2}\right).
    \end{equation}
\end{theorem}

Consequently, considering a depth $d (>d^*)$ random circuit $\mathcal{C}$ and another depth $d^*$ random circuit $\mathcal{C}'$ which consists of the last $d^*$ layers of $\mathcal{C}$, we have
\begin{equation}
    \mathbb{E}_{\mathcal{C}}\|\mathcal{C}(\ketbra{0^n}{0^n}) - \mathcal{C}'(\ketbra{0^n}{0^n})\|_1 =
    O\left(2^{n/2}\left(1-\frac{\gamma}{2}\right)^{n(d^*-1)/2}\right).
\end{equation}
Denoting $P$ and $Q$ as the output distributions of $\mathcal{C}(\ketbra{0^n}{0^n})$ and $\mathcal{C}'(\ketbra{0^n}{0^n})$, respectively, $\|P-Q\|_1$ is upper bounded by the above expression, and thus $\|P-Q\|_1$ can be made arbitrarily small by taking $d^* = O(1)$.

Finally, note that our two-qubit gate model interpolates between these extremes: single-qubit random circuits with $d^* = O(\log n)$ and global gate circuits with $d^* = O(1)$. It is therefore natural to expect that the effective depth of the two-qubit gate model lies between $O(1)$ and $O(\log n)$, and in particular is at most $O(\log n)$. We leave a rigorous proof of this claim for future work.

\end{document}